\documentclass[10pt,journal]{IEEEtran}

\ifCLASSOPTIONcompsoc

\usepackage[nocompress]{cite}
\else

\usepackage{cite}
\fi

\ifCLASSINFOpdf

\else

\fi

\usepackage{subfigure}
\usepackage{graphicx}
\usepackage{balance}
\usepackage{amsmath,amsthm, amssymb}
\usepackage{amssymb,wasysym}
\usepackage{amsfonts,balance}
\usepackage{mathrsfs}
\usepackage{cite}
\usepackage{color}
\usepackage{url}
\usepackage{bm}
\usepackage{wrapfig}
\usepackage{verbatim}
\usepackage{dsfont}
\usepackage{amsmath}
\usepackage{algorithm}
\usepackage{algpseudocode}
\usepackage{xurl}
 \usepackage[dvipsnames]{xcolor}
\usepackage{tabularx}
\usepackage{enumitem}
\usepackage{float}

\usepackage[colorlinks=true, linkcolor=black, citecolor=black, urlcolor=black, bookmarks=false]{hyperref}

\theoremstyle{definition}
\newtheorem{theorem}{\textbf{Theorem}}
\newtheorem{proposition}{\textbf{Proposition}}

\newtheorem{corollary}{\textbf{Corollary}}

\newtheorem{assumption}{\textbf{Assumption}}
\newtheorem{remark}{\textbf{Remark}}

\allowdisplaybreaks[3]

\ifodd 1
\else

\fi

\graphicspath{{figures/}}

\begin{document}

	\title{SiftMoE: Similarity-Aware Energy-Efficient Expert Selection for Wireless Distributed MoE Inference}
	
	\author{Qian Chen,~\IEEEmembership{Member,~IEEE}, Xianhao Chen,~\IEEEmembership{Member,~IEEE}, and Kaibin Huang,~\IEEEmembership{Fellow,~IEEE}
		\thanks{Q. Chen, X. Chen, and K. Huang are with the Department of Electrical and
Computer Engineering, The University of Hong Kong, Hong Kong (Email: qchen@eee.hku.hk, xcheneee@hku.hk, and huangkb@hku.hk). Corresponding authors: X. Chen and K. Huang.}
	
	}

	\markboth{Journal of \LaTeX\ Class Files}%
	{Shell \MakeLowercase{\textit{et al.}}: Bare Advanced Demo of IEEEtran.cls for IEEE Computer Society Journals}

	
		\IEEEtitleabstractindextext{
				\begin{abstract}
	Mixture-of-Experts (MoE) architectures leverage sparse activation to enhance the scalability of large language models (LLMs), making them suitable for deployment in resource-constrained edge networks. However, the sheer number of experts often exceeds the memory capacity of individual edge nodes, necessitating wireless distributed MoE (WIDE) inference where experts are spread across multiple edge nodes. In this context, expert selection directly affects communication costs. 
    Motivated by the similarity of experts, we propose SiftMoE, which judiciously selects or skips experts to strike a tradeoff between communication costs and inference accuracy. Specifically, we first establish theoretical bounds on the accuracy degradation resulting from expert replacement or skipping. Based on the bounds, we formulate an energy minimization problem for expert selection in WIDE inference subject to latency and accuracy constraints. In particular, for slow-fading channels, we derive optimal expert selection policies for both single-token decoding and multi-token prefilling. For fast-fading channels, we further extend our scheme to cope with rapidly varying channel conditions. Simulation results demonstrate that SiftMoE significantly reduces energy consumption while maintaining inference accuracy compared with conventional Top-K routing in WIDE systems.

					\end{abstract}

				}
	
	\maketitle

	\IEEEdisplaynontitleabstractindextext
	\IEEEpeerreviewmaketitle

\section{Introduction}
Large language models (LLMs) have demonstrated remarkable performance across a range of intelligent applications, including natural language processing, autonomous driving, healthcare, and embodied AI~\cite{zhao2023survey}. Motivated by stringent latency requirements and growing privacy concerns, pushing LLMs from cloud to the network edge is becoming a trend~\cite{10835069}. However, deploying such large models at the edge remains fundamentally challenging due to severe constraints on computation, memory, and energy resources of edge devices~\cite{zhengacmsurvey}.
To alleviate these limitations, Mixture-of-Experts (MoE) architectures have recently emerged as a promising paradigm~\cite{jiang2024mixtral}. In MoE models, each transformer layer replaces the conventional single feed-forward network (FFN) with multiple expert networks, activating only a small subset for each input token.
Due to their sparse activation patterns, MoE models enable substantial performance improvements while maintaining a computational budget comparable to that of conventional dense LLMs~\cite{gao2025mola}. This property makes MoE models particularly appealing for scalable and efficient LLM deployment in resource-constrained edge environments~\cite{11022729}.

Despite their advantages, a primary difficulty in MoE arises from the massive memory footprint incurred by the large number of experts, which often exceeds the memory capacity of a single edge node~\cite{10906629,10494556}. This limitation has motivated wireless distributed MoE (WIDE) frameworks, where experts are placed across multiple edge nodes to enable collaborative MoE inference~\cite{11395617}.
Unlike conventional MoE serving in data centers or GPU clusters, where experts are typically hosted by GPUs or servers connected through high-speed interconnects such as NVLink, InfiniBand, or RDMA, WIDE framework targets resource-constrained edge deployments where experts are distributed across the user device and nearby wireless helpers. Such a setting can exploit distributed computation and memory resources at the edge to support MoE models that cannot be hosted by a single edge node alone, while providing latency and privacy benefits by keeping inference close to the user.
However, once experts are distributed across the edge network, expert activation decisions become tightly coupled with wireless communications, as the selected experts determine the links traversed and hence the associated transmission costs~\cite{11162260,11143883}. Conventional Top-$K$ routing policies, which select $K$ highest-scoring experts per layer based solely on model-internal gating scores~\cite{fang2026hfedmoe}, neglect channel states and may route hidden states to experts over poor-quality links, leading to excessive communication energy consumption and latency. Furthermore, recent empirical studies on MoE inference reveal that Top-$K$ can be suboptimal, as lower-ranked experts can be replaced without much performance degradation~\cite{gupta2024lynx,lu2024not}, and experts within the same layer exhibit substantial functional overlap~\cite{li2024merge,li2025sub}. 
Taken together, these observations indicate that expert activation need not strictly follow conventional Top-$K$ routing and can instead be made more flexible, especially when expert activations incur significantly varied communication costs in distributed systems.

Existing studies have explored expert selection beyond fixed Top-$K$ activation, which can be broadly divided into two categories according to whether wireless communication is explicitly considered. The first category focuses on general MoE inference optimization, where expert selection is mainly used to reduce computation, memory-access, or expert-loading overhead while preserving inference accuracy~\cite{huang-etal-2024-harder,adapmoe_iccad,zhu2025enabling,11022729}. One line of work dynamically skips low-importance experts according to gating scores or sensitivity metrics, thereby avoiding unnecessary computation~\cite{huang-etal-2024-harder,adapmoe_iccad}. Another line targets the substantial CPU-GPU expert-loading overhead in practical MoE serving systems by replacing active experts with functionally similar resident experts or strategically skipping them~\cite{11022729,zhu2025enabling}.
Other related studies investigate quality-of-service (QoS)-aware LLM serving~\cite{imani2025qos}, temporal-aware routing mechanisms~\cite{11364210}, and precision-adaptive resource control schemes for MoE models~\cite{10937027}. Although these methods demonstrate the feasibility of flexible expert skipping and replacement beyond fixed Top-$K$ routing, they do not explicitly consider wireless communication costs or channel conditions. The second category, which is more directly related to our problem, considers expert selection in WIDE inference. WDMoE compares the router-weight vector with the latency vector after Top-$K$ routing and removes a low-weight expert when the similarity falls below a predefined threshold, thereby reducing communication costs~\cite{11083676}. Another work jointly accounts for task relevance and wireless channel states when selecting experts distributed across edge nodes, aiming to minimize transmission energy under layer-wise QoS constraints~\cite{qin2025optimal}. 
Nevertheless, as summarized in Table~\ref{tab:related_work_comparison}, existing expert-selection methods either overlook wireless channel effects or rely on heuristic metrics, such as cosine similarity, Hessian-based sensitivity, and routing scores, without providing theoretical selection-error analysis. Therefore, they remain insufficient for balancing the accuracy-communication tradeoff in WIDE expert selection.

\begin{table*}[t]
\centering
\caption{Comparison with representative expert-selection methods.}
\label{tab:related_work_comparison}
\footnotesize
\setlength{\tabcolsep}{3pt}
\renewcommand{\arraystretch}{1.1}
\resizebox{0.92\textwidth}{!}{
\begin{tabular}{c|cccccc}
\hline
\textbf{Work}
&
\begin{tabular}[c]{@{}c@{}}\textbf{Expert}\\\textbf{replacement}\end{tabular}
&
\begin{tabular}[c]{@{}c@{}}\textbf{Expert}\\\textbf{skipping}\end{tabular}
&
\begin{tabular}[c]{@{}c@{}}\textbf{Functional}\\\textbf{similarity}\end{tabular}
&
\begin{tabular}[c]{@{}c@{}}\textbf{Expert}\\\textbf{importance}\end{tabular}
&
\begin{tabular}[c]{@{}c@{}}\textbf{Wireless}\\\textbf{channel awareness}\end{tabular}
&
\begin{tabular}[c]{@{}c@{}}\textbf{Theoretical}\\\textbf{selection-error analysis}\end{tabular}
\\
\hline
\cite{huang-etal-2024-harder} & -- & $\checkmark$ & -- & $\checkmark$ & -- & -- \\
\cite{adapmoe_iccad}          & -- & $\checkmark$ & -- & $\checkmark$ & -- & -- \\
\cite{zhu2025enabling}        & $\checkmark$ & -- & $\checkmark$ & $\checkmark$ & -- & -- \\
\cite{11022729}               & $\checkmark$ & $\checkmark$ & $\triangle$ & $\checkmark$ & -- & -- \\
\cite{11083676}               & -- & $\checkmark$ & -- & $\checkmark$ & Slow fading & -- \\
\cite{qin2025optimal}         & -- & $\triangle$ & -- & $\checkmark$ & Slow fading & -- \\
\textbf{Proposed SiftMoE}     & $\checkmark$ & $\checkmark$ & $\checkmark$ & $\checkmark$ & Slow/fast fading & $\checkmark$ \\
\hline
\end{tabular}}
\vspace{2pt}

\begin{minipage}{\textwidth}
\small
$\checkmark$: explicitly considered;
$\triangle$: partially or indirectly considered;
--: not considered.
``Slow fading'' means that the small-scale fading variable is assumed to remain unchanged within each MoE layer, while ``fast fading'' means that the small-scale fading variable may vary within an MoE layer.
\end{minipage}
\end{table*}

To overcome the limitations above, this paper aims to answer a fundamental question: \textit{can expert selection be redesigned to strike the optimal balance between wireless communication costs and inference accuracy?} 
To this end, we propose a \underline{si}milarity-aware energy-e\underline{f}ficient expert selec\underline{t}ion framework for wireless distributed \underline{MoE} (SiftMoE) inference, which integrates theoretical selection-error analysis with channel-aware optimization.
Specifically, since selection-induced errors may accumulate and propagate across MoE layers with heterogeneous layer sensitivity~\cite{gupta2024lynx,lo2025closer,huang2025modes}, we first develop a theoretical framework to quantify the effect of expert replacement or skipping on inference accuracy.
Building on this analysis, SiftMoE exploits functional similarity among experts within each MoE layer to determine when experts should be replaced or skipped, thereby enabling energy-efficient inference under wireless fading while preserving accuracy.
The main contributions of this work are summarized as follows.

\begin{itemize}
    \item \textbf{Theoretical error analysis of expert selection}: We establish the \textit{first} theoretical analysis of how expert selection affects output deviations in WIDE inference.
 Our analysis reveals two fundamental insights. First, experts with small routing weights contribute marginally to the output, implying that removing or replacing them causes only limited distortion. Second, expert replacement is beneficial only when a functionally similar alternative exists; otherwise, directly skipping the expert may yield a tighter error bound. These results provide theoretical guidance for expert selection across layers under resource constraints.

    \item \textbf{Optimal expert selection under slow fading:}
Based on the analysis above, we formulate an energy minimization problem over expert selection that accounts for wireless fading, latency constraints, and accuracy guarantees induced by similarity-aware replacement or skipping. For the slow-fading scenario, where channel conditions remain unchanged upon propagating each MoE layer, we design efficient algorithms to obtain optimal expert selection schemes for single-token decoding and multi-token prefilling phases, respectively.

\item \textbf{Expert selection and adaptive transmission under fast fading:} For the fast-fading scenario, where channel states vary within each MoE layer, we reformulate the problem as minimizing the expected energy consumption subject to latency and accuracy constraints. Since jointly optimizing expert selection and transmission is difficult due to stochastic channel variations, we introduce a deterministic surrogate of the expected transmission energy to determine expert selection, whose structure is identical to the slow-fading problem. Then, we employ a dynamic programming (DP) approach to determine the optimal number of transmitted bits in each time slot.

\item \textbf{Experimental results:} We conduct extensive experiments using state-of-the-art MoE models and real-world datasets to evaluate the proposed framework under wireless edge deployment scenarios. Compared with conventional Top-$K$ routing, SiftMoE achieves higher inference accuracy while significantly reducing communication energy consumption under both slow- and fast-fading scenarios.

\end{itemize}

The remainder of this paper is organized as follows.
Section \ref{sec:system_model} presents the WIDE inference architecture together with the communication, computation, and energy models.
Section \ref{sec:error_analysis} establishes theoretical error bounds for expert selection.
Building on these results, Section \ref{sec:slow_fading} studies the energy-efficient expert selection scheme under slow fading.
Then, Section \ref{sec:fast_fading} extends the framework to fast-fading environments and develops a joint expert selection and adaptive transmission scheme.
Section \ref{sec:experiment} presents experimental results and Section \ref{sec:conclusion} concludes the paper.

\section{System Model}\label{sec:system_model}
As shown in Fig. \ref{fig:system_model}, we consider a WIDE inference system including one user and a set of $N-1$ helpers, denoted by $\mathcal{N}=\{1,\ldots,N-1\}$. 
Here, the user denotes the device generating the inference request, such as a mobile or vehicle-mounted device, while the helpers denote nearby edge servers or edge computing nodes assisting WIDE inference.
Let $\mathcal{V} \triangleq \{\mathrm{UE}\}\cup\mathcal{N}$ denote the set of all edge nodes, consisting of the user and all helpers, with $|\mathcal{V}|=N$.
The user generates an inference task composed of $M$ input tokens, indexed by the set $\mathcal{M} = \{1, \ldots, M\}$.
The inference task should be processed by an MoE model to generate the output results. The MoE model has $L$ MoE layers, indexed by $\mathcal{L} = \{1, \ldots, L\}$. Each MoE layer $\ell \in \mathcal{L}$ contains $N$ experts, and the expert set of layer $\ell$ is denoted by $\mathcal{E}^{(\ell)} = \{1^{(\ell)}, \ldots, N^{(\ell)}\}$. 
The experts are distributed across the edge nodes.
Due to GPU memory constraints and parallel processing requirements, we consider edge nodes with a single expert-execution unit, e.g., a single GPU, and assume that each edge node can load and execute at most one expert from each MoE layer at any given time~\cite{11224727,11083676,11162260}. This assumption is practical for both on-demand loading and preloaded serving in resource-constrained scenarios. Under on-demand loading, transferring one Mixtral-8$\times$7B expert of about 336 MB from CPU to GPU memory can take about 27.35 ms on an NVIDIA A100 GPU with PCIe Gen4~\cite{10937027}, whereas transmitting one BF16 hidden state of Mixtral-8$\times$7B with dimension 4096 and returning an expert output of the same dimension take about 6.55 ms and 1.31 ms under representative uplink and downlink rates of 10 Mbps and 50 Mbps, respectively. Thus, loading a single expert may be much slower than the corresponding uplink-downlink transmission of hidden states. Under preloaded serving, the assumption also avoids scheduling and computation contention among the same-layer experts on the same edge node\footnote{Practical deployments may allow one edge node to host multiple experts from the same MoE layer, especially when the node is equipped with multiple GPUs. The proposed SiftMoE framework can be extended to this case by viewing each GPU as a logical helper that hosts one expert. Logical helpers on the same physical edge node share the same wireless link to the user, while still executing their assigned experts separately. Under this interpretation, the current one-to-one expert-node mapping provides a conservative communication model and yields an upper-bound estimate of total energy consumption.}. In addition, different edge nodes host distinct experts within the same layer, so selecting an expert in layer $\ell$ uniquely determines the serving edge node. This one-to-one mapping allows us to focus on the key coupling among expert selection, wireless channel states, and communication costs.

As shown in Fig. \ref{fig:pipeline}, at each MoE layer, the attention and gating networks are first executed at the user to determine which experts are activated. 
This is feasible for resource-constrained users because the gating network is lightweight and the attention module is much smaller than the activated expert networks in each MoE layer. It also helps preserve privacy by keeping the raw input, layer-wise aggregated representations, and final inference results at the user side.
Since attention and gating operations incur a fixed computational workload per layer, their processing latency and energy consumption are treated as constants and excluded from the subsequent analysis.
If a selected expert resides on a helper rather than the user, the hidden states are exchanged between the user and the helper. We use $b$ to denote the data size of the transmitted hidden state per token.

\begin{figure}[t]
	\centering
    \includegraphics[width = 0.5\textwidth]{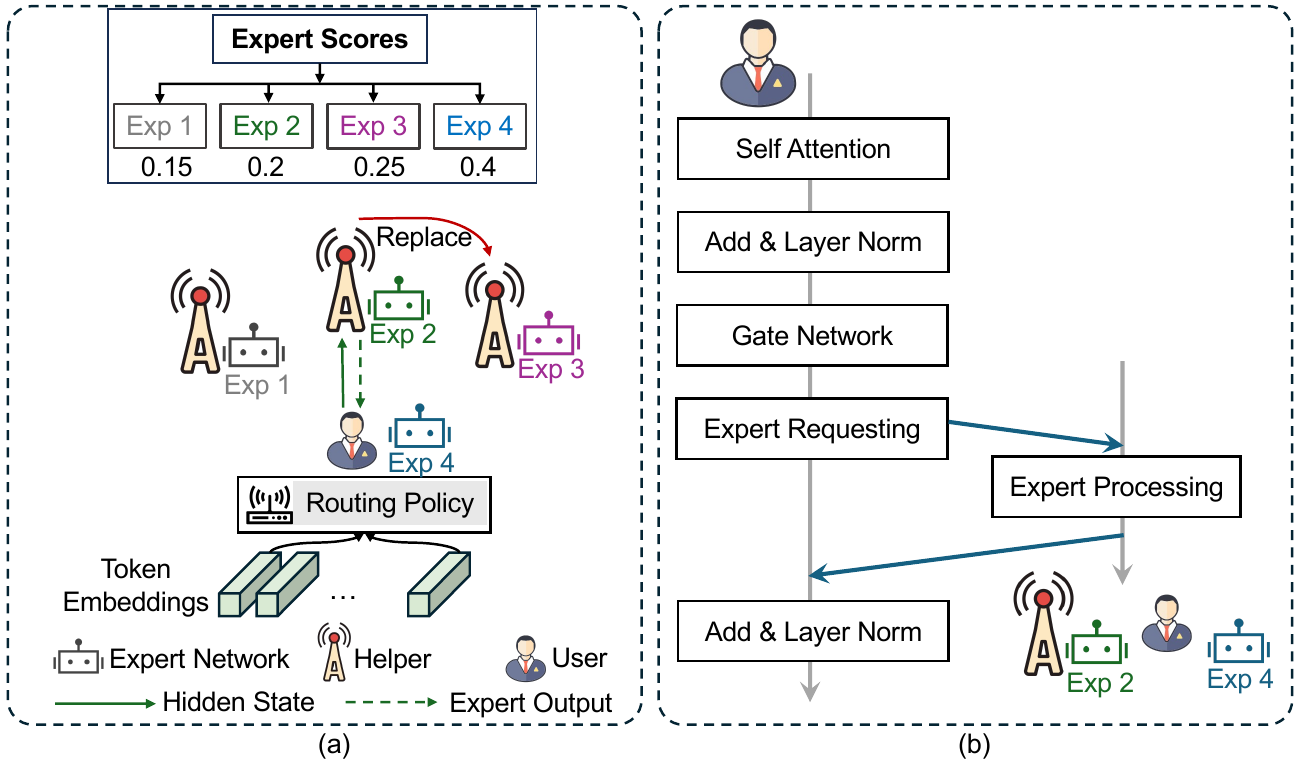}
\caption{An illustration of the WIDE inference system, where a user sends its tokens to the helpers with the required experts for processing. (a) The proposed SiftMoE framework, where ``Exp" denotes expert. Here, expert 2 has a similar function but better channel condition compared with expert 3. (b) Operations of inference within an MoE layer.
\label{fig:system_model}}
\end{figure}

\begin{figure*}
    \centering
    \includegraphics[width=0.9\textwidth]{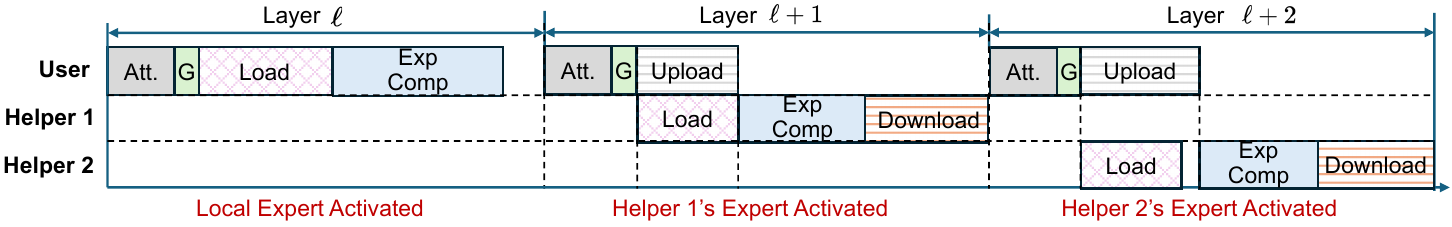}
    \caption{Timeline of SiftMoE inference.}
    \label{fig:pipeline}
\end{figure*}

\subsection{Local Computing Model}
We first consider the case where a token activates local experts at the user. In this case, the inference consists of two stages: i) loading the expert into the user's GPU memory with latency $T_{\mathrm{UE,load}}$, and ii) performing local expert computing.

Let $x_{m,\mathrm{UE}}^{(\ell)}$ be a binary variable indicating whether token $m$ is processed by the local expert of MoE layer $\ell$.
Specifically, $x_{m,\mathrm{UE}}^{(\ell)}=1$ if token $m $ activates the local expert at layer $\ell$, and $x_{m,\mathrm{UE}}^{(\ell)}=0$ otherwise.
Then, the number of tokens processed locally at layer $\ell$ is given by $D_{\mathrm{UE}}^{\left( \ell \right)} = \sum_{m \in \mathcal{M}}x_{m,\mathrm{UE}}^{\left( \ell \right)}$, and the corresponding computing latency is 
\begin{equation}
T_{\mathrm{UE,cp}}^{(\ell)}=\frac{D_{\rm{UE}}^{(\ell)}\phi}{C_{\rm{UE}}},
\end{equation}
where $\phi$ is the expert computing workload (in FLOPs) for processing a hidden state 
and $C_{\rm{UE}}$ is the computation capability (in FLOP/s) of the user.

Accordingly, the total latency at MoE layer $\ell$ when activating local expert is given by
\begin{equation}
    T_{\mathrm{UE}}^{(\ell)}  = I_{\rm{UE}}^{(\ell)}T_{\mathrm{UE,load}} + T_{\mathrm{UE,cp}}^{(\ell)}.
\end{equation}
Here, $I_{\rm{UE}}^{(\ell)}  = \mathbb{I}\left (D_{\rm{UE}}^{(\ell)}  \right ) $ is an indicator function, where $I_{\rm{UE}}^{(\ell)}=1$ only if $D_{\rm{UE}}^{(\ell)}>0$. Otherwise, $I_{\rm{UE}}^{(\ell)}=0$.

Let $P_{\rm{UE,load}}$ denote the power of loading an expert at the user and $P_{\rm{UE,cp}}$ denote the local computing power. Then, the energy consumption incurred by activating the local expert at layer $\ell$ is given by
\begin{equation}\label{equ:energy_local}
    E_{\rm{UE}}^{(\ell)} = I_{\rm{UE}}^{(\ell)}P_{\rm{UE,load}}T_{\rm{UE,load}}+ P_{\rm{UE,cp}}T_{\rm{UE,cp}}^{(\ell)}.
\end{equation}

\subsection{Computation Offloading Model}
We next consider the case where a token activates an expert at a helper $n \in \mathcal{N}$.
In this case, the inference of each MoE layer consists of three stages: i) uplink transmission of the hidden state from the user to the helper, ii) expert loading and computation at the helper, and iii) downlink transmission of the output back to the user. Loading an expert into the helper $n$'s GPU memory incurs an additional latency, denoted by $T_{n,\mathrm{load}}$, which can be performed in parallel with the uplink transmission.

We assume the total bandwidth $B$ assigned to the user is partitioned among the helpers, where $B_n$ denotes the bandwidth allocated to helper $n$. 
After determining the activated expert(s) of each token, the user simultaneously uploads its hidden state to the helpers hosting those experts.
Let $p_{{\rm{UE}},n}^{(\ell)}$ denote the transmit power from the user to helper $n$. Then, the corresponding uplink achievable rate $R_{{\mathrm{UE}},n}^{(\ell)}$ is given by
	\begin{equation}\label{equ:uplink_rate}
		R_{{\mathrm{UE}},n}^{(\ell)} = B_n \log_2\left ( 1+ \frac{p_{{\rm{UE}},n}^{(\ell)}   d_n^{-\alpha}G_{\mathrm{UE},n} h_{n}^{(\ell)} }{N_0B_n}  \right ),
	\end{equation}
    where $d_n$ is the distance between the user and helper $n$\footnote{We assume that the user position remains unchanged during one inference request, corresponding to scenarios where the inference latency is shorter than the user mobility timescale. If noticeable movement occurs within one request, the framework can be extended by replacing $d_n$ with a layer-dependent distance $d_n^{(\ell)}$.}, $\alpha$ is the path-loss factor, $G_{\mathrm{UE},n}$ is the uplink antenna-related factor,
    $N_0$ is the noise power spectral density, and $h_{n}^{(\ell)}$ is the small-scale channel fading variable.

    Similarly, after processing, the expert output should be downloaded from the helpers to the user. The downlink achievable rate $R_{n, {\mathrm{UE}}}^{(\ell)}$ can be expressed as 
    \begin{equation}
        R_{n, {\mathrm{UE}}}^{(\ell)}  = B_n  \log_2\left ( 1+ \frac{P_{n,\mathrm{UE}} d_n^{-\alpha} G_{n,{\mathrm{UE}}} h_{n}^{(\ell)}  }{N_0 B_n }  \right ),
    \end{equation}
where $P_{n,\mathrm{UE}}$ denotes the transmit power from helper $n$ to the user and $G_{n,{\mathrm{UE}}}$ is the downlink antenna-related factor.

Let $x_{m,n}^{(\ell)}$ denote a binary expert selection variable at the helpers, where $x_{m,n}^{(\ell)}=1$ indicates that token $m $ activates the expert at helper $n$ in MoE layer $\ell$, and $x_{m,n}^{(\ell)}=0$ otherwise.
Accordingly, the number of tokens processed by helper $n$ at layer $\ell$, denoted by $D_n^{\left( \ell \right)} $, is obtained by  $D_n^{\left( \ell \right)} = \sum_{m \in \mathcal{M}}x_{m,n}^{\left( \ell \right)}$.
Therefore, the uplink and downlink transmission latency between the user and helper $n$ for layer $\ell$ is calculated as 
\begin{equation}\label{equ:uplink_latency}
T_{{\mathrm{UE}},n}^{(\ell)}=\frac{D_n^{(\ell)}b}{R_{{\mathrm{UE}},n}^{(\ell)}},
\end{equation}
and
\begin{equation}\label{equ:downlink_latency}
T_{n,{\mathrm{UE}}}^{(\ell)}=\frac{D_n^{(\ell)}b}{R_{n,{\mathrm{UE}}}^{(\ell)}}.
\end{equation}
Moreover, expert processing latency at helper $n$ during layer $\ell$ is given by
\begin{equation}
T_{n,\mathrm{cp}}^{(\ell)}=\frac{D_n^{(\ell)}\phi}{C_n},
\end{equation}
where $C_n$ is the computation capability (in FLOP/s) of helper $n$.
Therefore, the total latency\footnote{Additional helper-side processing overheads, such as kernel-launch overhead, memory-management overhead, and synchronization, can be incorporated as a fixed helper- and layer-dependent constant latency term. In practice, this constant term can be subtracted from the tolerable latency of the corresponding layer when checking the latency feasibility of that helper. Since it does not change the structure of the proposed expert-selection problem, it is omitted for notational simplicity.} incurred by activating the expert at helper $n$ during MoE layer $\ell$ is given by
\begin{equation}
    T_{n}^{(\ell)} =  \max \left\{  T_{{\mathrm{UE}},n}^{(\ell)}, T_{n, \rm{load}}\right\}
     + T_{n,\mathrm{cp}}^{(\ell)}+T_{n,{\mathrm{UE}}}^{(\ell)}.
\end{equation}

By combining (\ref{equ:uplink_rate}) and (\ref{equ:uplink_latency}), the transmission energy consumption from the user to helper $n$ for MoE layer $\ell$, denoted by $E_n^{(\ell)}$, is given by
\begin{equation}\label{equ:energy-consumption-UE-n}
E_n^{(\ell)}=p_{{\rm{UE}},n}^{(\ell)}T_{{\mathrm{UE}},n}^{(\ell)}
	=
	\frac{\left ( 2^{\frac{D_n^{(\ell)}b}{B_nT_{{\mathrm{UE}},n}^{(\ell)}}} -1 \right )N_0 B_nT_{{\mathrm{UE}},n}^{(\ell)}}{d_n^{-\alpha}G_{\mathrm{UE},n}h_{n}^{(\ell)}}.
\end{equation}

We ignore the receiving power consumption at the user because it is often much smaller than the transmit power~\cite{6884852,9210812,9461628}.

\section{Error Analysis of Expert Selection}\label{sec:error_analysis}
In this section, we first provide a theoretical analysis of the per-layer error caused by expert selection. We then derive the allowable error at each MoE layer under a constraint on the final output deviation. These results lay the analytical foundation for the energy-efficient expert selection design developed in Sections \ref{sec:slow_fading} and \ref{sec:fast_fading}.

\subsection{Preliminaries and Assumptions}
Consider an input token representation $z \in \mathbb{R}^{d}$, where $d$ denotes the hidden-state dimension.
At each MoE layer, the gating network computes a vector of routing logits given by ${\mathcal{P}}\left ( z \right )   =  \mathbf{W}_r ^{\mathrm{T} }z $,
where $\mathbf{W}_r \in \mathbb{R}^{d \times N}$ is the router matrix. 
Based on the routing logits, an expert activation set $\mathcal{A}$ is selected. Thus, the output of the MoE layer can be expressed as $y\left ( z  \right )   = \sum_{i \in {\mathcal{A}}} {g_{i} \left ( z  \right ) {\mathrm{FFN}}_{i}\left ( z  \right ) }$,
where ${\mathrm{FFN}}_{i}\left ( z  \right )$ is the output result of expert $i$ and $g_{i} \left ( z  \right )  = \frac{\exp\left ( {\mathcal{P}}_i\left ( z \right ) \right ) }{\sum_{i' \in \mathcal{A} }\exp\left ( {\mathcal{P}}_{i'}\left ( z \right ) \right )}$ denotes the softmax-normalized gating score associated with expert $i$.

Each expert is implemented as an FFN, which is a standard neural network architecture. Therefore, we adopt the following assumptions commonly used in the theoretical analysis of deep neural networks.

\begin{assumption}(Bounded expert outputs)\cite{JMLR_nonasym}\label{assumption:bounded_output}
For the $\ell$-th MoE layer, the outputs of expert networks are assumed to be uniformly bounded over $ \mathcal Z$, i.e., $\left\|{\rm FFN}_i^{(\ell)}(z)\right\|_2 \le B_{\max}^{(\ell)}$ for $z\in\mathcal Z$ and $i\in \mathcal{E}^{(\ell)}$.
\end{assumption}

\begin{assumption}(Lipschitz continuity of experts)~\cite{lipschitz_nips,mohri2007stability}\label{assumption:Lip_expert}
For the $\ell$-th MoE layer, each expert network ${\rm FFN}_i^{(\ell)}(\cdot)$ is assumed to be $\beta_{{\rm E},i}^{(\ell)}$-Lipschitz continuous over $\mathcal{Z}$, i.e., $\left\| {\rm FFN}_i^{(\ell)}(\tilde z)-{\rm FFN}_i^{(\ell)}(z)\right\|_2
\le
\beta_{{\rm E},i}^{(\ell)}\left\|\tilde z-z\right\|_2$ for $z, \tilde{z} \in \mathcal{Z}$ and $i\in \mathcal{E}^{(\ell)}$.
\end{assumption}

\subsection{Effect of Expert Selection on Output Results}

Consider MoE layer $\ell$. Let $\mathcal{A}_{\rm{Top}}^{(\ell)}\left ( z \right ) $ denote the set of activated expert(s) at layer $\ell$ with an input hidden state $z$ under traditional Top-$K$ selection strategy. The resulting output representation at this layer under the Top-$K$ scheme is given by
\begin{equation}
    y_{\rm{Top}}^{(\ell)} \left ( z  \right )
		=\sum_{i \in \mathcal{A}_{\rm{Top}}^{(\ell)}\left ( z \right ) } g_i^{(\ell)}\left ( z  \right ) {\mathrm{FFN}}_i^{(\ell)}\left ( z  \right ),
\end{equation}
    where $g_i^{(\ell)}(z)$ denotes the gating score associated with expert $i$ at layer $\ell$.

   When resource-aware expert selection is applied, the resulting activated expert set is denoted by $\tilde{\mathcal{A}}^{(\ell)}\left ( z \right ) $. In this case, the layer output becomes
    \begin{equation}
        \tilde{y}^{(\ell)} \left ( z  \right )
		=
		\sum_{i \in {\mathcal{A}}_{\rm{Top}}^{(\ell)}\left ( z \right ) } g_i^{(\ell)}\left ( z  \right ) {\mathrm{FFN}}_{\varphi(i)}^{(\ell)}\left ( z   \right ).
    \end{equation}
Here, for any expert $i \in \mathcal{A}_{\rm{Top}}^{(\ell)}\left ( z \right )$, $\varphi(i) $ denotes the expert actually used under the resource-aware expert selection scheme. 
If $\varphi(i)= i$, expert $i$ is retained. If $\varphi(i) \neq i$ and $\varphi(i) \neq \emptyset$, then $\varphi(i) $ is a substitute expert for $i$. Otherwise, $\varphi(i) = \emptyset$, indicating that expert $i$ is skipped.

Define the deviation induced by expert selection at layer $\ell$ with input $z$ as $\delta^{(\ell)}\left ( z  \right ) \triangleq \left \| \tilde{y}^{(\ell)}\left ( z  \right ) - y_{\rm{Top}}^{(\ell)}\left ( z  \right ) \right \|_2$. The following proposition establishes an upper bound on the deviation at each MoE layer.

\begin{proposition}[Layer-wise deviation]\label{prop:layer-expect-perturbation}
    With expert selection, for any input hidden state $z$ at MoE layer $\ell$, the deviation satisfies
\begin{equation}\label{equ:prop1_layerwise}
 \delta^{(\ell)}\left ( z  \right ) \leq \sum_{i \in {\mathcal{A}}_{\rm{Top}}^{(\ell)}\left ( z   \right )} g_i^{(\ell)}\left ( z  \right ) F_{i,\varphi(i)}^{(\ell)}\left ( z   \right ) ,
\end{equation}
where $F_{i,\varphi(i)}^{(\ell)}\left ( z   \right )\triangleq \left \| {\rm{FFN}}_{i}^{(\ell)}\left ( z   \right ) \right \|_2  \cdot \sqrt{1+\rho_{i,\varphi(i)}^{(\ell)}\left ( z   \right )^2-2\rho_{i,\varphi(i)}^{(\ell)}\left ( z   \right ) \cos \theta_{i,\varphi(i)}^{(\ell)}\left ( z   \right ) }$. Here, $\rho_{i,\varphi(i)}^{(\ell)}\left ( z   \right ) = \frac{\left\|\mathrm{FFN}_{\varphi(i)}^{(\ell)}\left ( z   \right )\right\|_2}{\left\|\mathrm{FFN}_{i}^{(\ell)}\left ( z   \right )\right\|_2} $ is the ratio between the output norms of experts $i$ and $\varphi(i)$, and $\cos \theta_{i,\varphi(i)}^{(\ell)}\left ( z   \right )=\frac{\mathrm{FFN}_{i}^{(\ell)}\left ( z   \right ) \cdot \mathrm{FFN}_{\varphi(i)}^{(\ell)}\left ( z   \right )}{\left\|\mathrm{FFN}_{i}^{(\ell)}\left ( z   \right )\right\|_2\left\|\mathrm{FFN}_{\varphi(i)}^{(\ell)}\left ( z   \right )\right\|_2}$ denotes the cosine similarity between their output vectors.  If $\varphi(i)=\emptyset$, expert $i$ is skipped, which is equivalent to replacing it with a null expert whose output is ${\mathrm{FFN}}_{\emptyset}^{(\ell)}(z)=\mathbf{0}$.
\end{proposition}

\begin{proof}
    Please see Appendix \ref{proof:prop-layer-expect-perturbation}.
\end{proof}

Proposition~\ref{prop:layer-expect-perturbation} indicates that experts with small gating scores contribute marginally to the aggregated output. Consequently, replacing or skipping such experts only induces a marginal change in the deviation. In contrast, experts with large gating scores play a dominant role in the layer output. Therefore, substituting a high-score expert requires selecting a functionally similar alternative. Otherwise, a large discrepancy between expert outputs may occur, thereby increasing the resulting layer-wise deviation.

\begin{remark}[Skipping versus replacement]
Based on Proposition~\ref{prop:layer-expect-perturbation}, when expert $i$ at layer $\ell$ is skipped,
the corresponding deviation term is $g_{i}^{(\ell)}\left ( z   \right ) \left \| {\rm{FFN}}_{i}^{(\ell)}\left ( z   \right ) \right \|_2 $. 
On the other hand, when expert $i$ is replaced by an alternative $\tilde i$ at layer $\ell$, the corresponding deviation term is $g_{i}^{(\ell)}\left ( z   \right )\left \| {\rm{FFN}}_{i}^{(\ell)}\left ( z   \right ) \right \|_2\sqrt{1+\rho_{i,\tilde i}^{(\ell)}(z)^2
-2\rho_{i,\tilde i}^{(\ell)}(z)\cos\theta_{i,\tilde i}^{(\ell)}(z)}$. When $\sqrt{1+\rho_{i,\tilde i}^{(\ell)}(z)^2
-2\rho_{i,\tilde i}^{(\ell)}(z)\cos\theta_{i,\tilde i}^{(\ell)}(z)} >1$, the deviation induced by replacement is larger than that caused by skipping. Such a case naturally arises when the substitute expert is functionally dissimilar to the original one. Consequently, skipping a low-impact expert may perform better than replacing it with a poorly matched alternative.
\end{remark}

\begin{remark}[Offline-online combined estimation of layer-wise deviation]\label{remark:offline_online}
Proposition~\ref{prop:layer-expect-perturbation} shows that the layer-wise deviation depends on two terms: the input-dependent gating score $g_i^{(\ell)}\left ( z   \right )$ and the input-dependent expert-pair mismatch $F_{i,\varphi(i)}^{(\ell)}\left ( z   \right )$. During inference, ${\mathcal A}_{\rm Top}^{(\ell)}\left ( z   \right )$ and $g_i^{(\ell)}\left ( z   \right ) $ are naturally obtained online when the current hidden state reaches the $\ell$-th MoE layer. However, the exact value of $F_{i,\varphi(i)}^{(\ell)}\left(z\right)$ cannot be obtained before expert execution, because it requires feeding the current hidden state $z$ into both the original expert $i$ and the substitute expert $\varphi(i)$ and then comparing their output vectors. This is inconsistent with the purpose of expert selection, which aims to determine whether an expert should be retained, replaced, or skipped before executing the corresponding experts. 

Therefore, we estimate the expert-pair mismatch offline using a calibration dataset and use it to construct an offline-calibrated deviation surrogate for online expert selection. Specifically, during online inference, we use
\begin{equation}\label{equ:offline_online_layerwise}
    \bar{\delta}^{(\ell)}(z)
=
\sum_{i\in{\mathcal A}_{\rm Top}^{(\ell)}(z)}
g_i^{(\ell)}(z)
\bar{F}_{i,\varphi(i)}^{(\ell)}.
\end{equation}
Here, ${\mathcal A}_{\rm Top}^{(\ell)}(z)$ and $g_i^{(\ell)}(z)$ are obtained online when the current hidden state $z$ reaches the $\ell$-th MoE layer, while $\bar{F}_{i,\varphi(i)}^{(\ell)}$ is an offline-calibrated expert-pair mismatch. 
Specifically, 
let $Z^{(\ell)}$ denote the random hidden state entering the $\ell$-th MoE layer and $\mathcal{D}_{\rm cal}^{(\ell)}$ denote the corresponding calibration distribution of hidden states at layer $\ell$.
For each layer $\ell$ and expert pair $\left(i,\varphi(i)\right)$, we define $\bar{F}_{i,\varphi(i)}^{(\ell)} \triangleq
\mathbb{E}_{Z^{(\ell)}\sim\mathcal{D}_{\rm cal}^{(\ell)}}
\left[
F_{i,\varphi(i)}^{(\ell)}\left(Z^{(\ell)}\right)
\right]$, which represents the average functional mismatch between expert $i$ and its substitute $\varphi(i)$ at layer $\ell$ over the calibration hidden-state distribution.
During online inference, computing the estimated layer-wise deviation only requires scalar lookup, multiplication, summation, and comparison with the layer-wise deviation budget. The storage overhead is also limited, since storing all ordered expert-pair mismatches requires only $\mathcal{O}(LN^{2})$ scalar values.
 This design preserves token-dependent expert importance while keeping the online expert-selection process lightweight\footnote{Since the expert parameters are fixed during inference, the offline-calibrated expert-pair mismatches can be reused across inference requests as long as the hidden-state distribution during inference remains close to that of the calibration dataset, and recalibration is only needed when the MoE model is updated or this distribution changes significantly.}.
\end{remark}

We next characterize how per-layer deviations induced by expert selection accumulate across layers and affect the final model output.
Let $e^{(\ell)} =	 \left \| \tilde{y}^{(\ell)} -y_{\rm{Top}}^{(\ell)}  \right \|_2$ denote the output deviation at layer $\ell$ with the initialization $e^{(0)}=0$, which captures the \textit{accumulated deviation} between the expert-selection output and the corresponding Top-$K$ output up to layer $\ell$.

\begin{theorem}[Final output deviation bound]\label{prop:final_deviation_perlayer}
    Given an expert selection scheme, there is an upper bound on the accumulated deviation of the model output space:
    \begin{equation}\label{equ:e-ell-upper}
        e^{(\ell)} \le \sum_{r=1}^{\ell}\left [ \left(2 B_{\max}^{(r)}+\delta^{(r)}\right)\prod_{t=r+1}^{\ell} \beta_{{\rm E},\max}^{(t)} \right ],
    \end{equation}
    where $ \beta_{{\rm E},\max}^{(\ell)} = \max_{i \in \mathcal{E}^{(\ell)} } \beta_{{\rm{E}},i}^{(\ell)}$.
\end{theorem}

\begin{proof}
Please see Appendix \ref{proof:prop-final_deviation_perlayer}.
\end{proof}

Theorem \ref{prop:final_deviation_perlayer} provides several insights into how layer-wise expert selection errors propagate across MoE layers. First, the final output deviation is determined by the cumulative effect of deviations introduced at all preceding layers, indicating that expert-selection errors are inherently cross-layer coupled. Second, the product term $\prod_{t=r+1}^{\ell} \beta_{{\rm E},\max}^{(t)}$ shows that deviations introduced at earlier layers may be amplified by subsequent layers, making shallow-layer errors generally more critical than those in deeper layers. Third, by separating the layer-wise selection error $\delta^{(r)}$ from the model-dependent constants $\beta_{{\rm E},\max}^{(\ell)}$ and $B_{\max}^{(\ell)}$, the bound reveals that different layers can tolerate different levels of deviation depending on their error amplification effects. Therefore, this result establishes a direct link between local expert-selection deviations and the final output deviation, enabling global accuracy requirements to be translated into tractable layer-wise constraints for the subsequent expert selection design.

Based on Theorem \ref{prop:final_deviation_perlayer}, we next characterize how the expected layer-wise deviation induced by expert selection should be controlled in order to guarantee a target expected output accuracy.

\begin{corollary}[Expected per-layer deviation budget]\label{corollary_budget}
Assuming that the expected intermediate deviation satisfies
$\mathbb{E}[e^{(\ell)}]\le \kappa^{(\ell)}$ for $\ell=1,\ldots,L-1$, and the final deviation satisfies $\mathbb{E}\!\left[e^{(L)}\right]\le \theta$ with a nonnegative allocation $\{\theta^{(\ell)}\}_{\ell=1}^L$ such that $\sum_{\ell=1}^{L}\theta^{(\ell)}=\theta$. 
Then, the expected deviation induced by expert selection at MoE layer $\ell$ must satisfy \label{acc_budget}
\begin{equation}
\mathbb{E}\left[\delta^{(\ell)}\right] \leq \eta^{(\ell)},
\end{equation}
where $\eta^{(\ell)} = \min\left \{\kappa^{(\ell)}- 2 B_{\max}^{(\ell)}-\sum_{r=1}^{\ell-1} f(r,\ell)   ,  \right.$ $\left.\frac{\theta^{(\ell)}}{\prod_{t=\ell+1}^{L} \beta_{{\rm E},\max}^{(t)}}- 2 B_{\max}^{(\ell)}\right \}$ and $f(r,\ell) =\left(2 B_{\max}^{(r)}+\delta^{(r)}\right)\prod_{t=r+1}^{\ell} \beta_{{\rm E},\max}^{(t)}$.
\end{corollary}

The allowable per-layer deviation budgets in Corollary \ref{corollary_budget} are obtained according to the target end-to-end (E2E) accuracy requirement and the model-dependent constants in Assumptions \ref{assumption:bounded_output} and \ref{assumption:Lip_expert}, which can be estimated offline on a calibration dataset using lightweight finite-difference tests~\cite{11314800,wei2025optimizing}. Once these budgets are determined, the online expert-selection procedure only checks whether the estimated layer-wise deviation satisfies the corresponding precomputed budget.

\section{Expert Selection Under Slow Fading}\label{sec:slow_fading}
In this section, we formulate an optimization problem for energy-efficient expert selection under slow-fading channels, leveraging the allowable per-layer deviations derived in Section \ref{sec:error_analysis} to enforce accuracy constraints. Under slow-fading channels, the small-scale fading coefficients may vary across different MoE layers but remain constant within each layer. We first study the decoding phase, where a single token is processed and the optimal solution admits a simple structure. Then, we extend the analysis to the multi-token prefilling phase. The fast-fading scenario will be discussed in Section \ref{sec:fast_fading}.

\subsection{Problem Formulation}
Our objective is to minimize the total energy consumption across all MoE layers by optimizing the expert selection decisions during inference.
Let $\mathbf{X}^{(\ell)}=\left \{x_{m,v}^{(\ell)}\right \}_{m \in \mathcal{M}, v \in \mathcal{V}}$ 
denote the binary expert selection variables at MoE layer $\ell$. Due to the sequential nature of MoE inference, the gating scores and expert activation information are revealed layer by layer, preventing us from solving the problem for the whole MoE model. Moreover, the energy consumption, as well as the latency and accuracy constraints, can be layer-wise separable if we introduce an accuracy budget for each layer according to Corollary \ref{acc_budget}. Therefore, to minimize the total energy, we focus on minimizing energy for each layer:
\begin{subequations}
	\begin{equation}
\mathcal{P}1: \quad		\min_{\mathbf{X}^{(\ell)}} \; E^{(\ell)} = E_{\rm{UE}}^{(\ell)}+ \sum_{n \in \mathcal{N}}E_n^{(\ell)}
	\end{equation}
\begin{equation}\label{constraint:P1_latency}
		{\rm{s.t.}} \quad T_{v}^{(\ell)} \leq T, \; \forall v \in \mathcal{V},  
	\end{equation}
\begin{equation}\label{constraint:P1_accuracy_per_layer}
		\mathbb{E} \left [ \left \| \delta^{(\ell)} \left ( z_m  \right )\right \|_2   \right ]  \leq \eta^{(\ell)} ,    \forall   m\in \mathcal{M},  
	\end{equation}
\begin{equation}\label{constraint:P1_num_active_expert}
		1 \leq \sum_{v \in \mathcal{V}}x_{m,v}^{(\ell)} \leq K, \forall m \in \mathcal{M},
	\end{equation}
\begin{equation}\label{constraint:P1_binary}
	x_{m,v}^{(\ell)} \in \left \{ 0,1 \right \}, \forall m \in \mathcal{M}, v \in \mathcal{V}, 
	\end{equation}
\end{subequations}
where Constraint (\ref{constraint:P1_latency}) ensures that the latency of each layer does not exceed a threshold $T$. 
Constraint (\ref{constraint:P1_accuracy_per_layer}) limits the inference accuracy degradation caused by expert selection at each MoE layer. Constraint (\ref{constraint:P1_num_active_expert}) restricts the number of experts activated by each token to at most $K$, ensuring fairness and comparability with the conventional Top-$K$ routing scheme.
Constraint (\ref{constraint:P1_binary}) denotes that the expert selection variable is a binary one.

\subsection{Decoding Phase with a Single Token}
During the decoding phase, only a single token is processed at each step of autoregressive inference. In this scenario, the token index $m$ can be omitted for simplicity. The resulting optimization problem admits a simplified structure and a more tractable solution.
Accordingly, the energy consumption of local computing can be expressed as
\begin{equation}
    {\bar E}_{\rm{UE}}^{(\ell)} = x_{{\rm{UE}}}^{(\ell)} \left(P_{\rm{UE,load}}T_{\rm{UE,load}}+ \frac{P_{\rm{UE,cp}}\phi}{C_{\rm{UE}}} \right).
\end{equation}

For helper-side processing, we notice that $E_n^{(\ell)}$ is a non-increasing function related to $T_{{\mathrm{UE}},n}^{(\ell)}$.
Thus, the uplink transmission time  $T_{{\mathrm{UE}},n}^{(\ell)}$ should be maximized if the expert at helper $n$ is activated, and
the optimal total latency at helper $n$ is $T_{n}^{(\ell) \ast}   = x_n^{(\ell)}T$. 
By jointly considering constraints (\ref{constraint:P1_latency}) and (\ref{constraint:P1_accuracy_per_layer}), we define the feasible set $\Omega^{(\ell)}$ as the collection of edge-node subsets, where each subset consists of at most $K$ edge nodes hosting experts that satisfy the latency and accuracy constraints.
Then, 
the optimal uplink transmission energy from the user to helper $n \in \Omega^{(\ell)}$ is given by
\begin{equation}
   {\bar E}_n^{(\ell)}  	 =  x_n^{(\ell)}
	\frac{\left ( 2^{\frac{b}{B_n {\bar T}_{{\rm{UE}},n}^{(\ell)} }} -1 \right )N_0 B_n {\bar T}_{{\rm{UE}},n}^{(\ell)}}{d_n^{-\alpha}G_{\mathrm{UE},n}h_{n}^{(\ell)}}.
\end{equation}
where ${\bar T}_{{\rm{UE}},n}^{(\ell)}=T - \left ( \frac{\phi}{C_n} + \frac{b}{R_{n,{\mathrm{UE}}}^{(\ell)}} \right )$ denotes the remaining time budget for uplink transmission.

Based on the above analysis, the following theorem characterizes the optimal solution to problem $\mathcal{P}1$ under the single-token decoding case.

\begin{theorem}
  Under the single-token decoding scenario, the optimal solution to problem $\mathcal{P}1$ at layer $\ell$ is expressed as
  \begin{equation}
   x_v^{(\ell)\ast} =
\begin{cases}
1, & v \in \mathcal J^{(\ell)\ast},\\
0, & \text{otherwise}.
\end{cases} 
\end{equation}
where 
  \begin{equation}
  \mathcal J^{(\ell)\ast} = \operatorname*{arg\,min}_{\mathcal{J}^{(\ell)} \in \Omega^{(\ell)}}
\sum_{v \in \mathcal{J}^{(\ell)}} f_v^{(\ell)},
\end{equation}
with $f_v^{(\ell)}$ being the energy cost at layer $l$ associated with selecting the expert hosted at edge node $v$, given by
\begin{equation}
    f_v^{(\ell)}=\begin{cases}
    P_{v,\rm{load}}T_{v,\rm{load}}+ \frac{P_{v,\rm{cp}}\phi}{C_v},  & \text{ if } v= \mathrm{UE}, \\
\frac{\left ( 2^{\frac{b}{B_v {\bar T}_{{\rm{UE}},v}^{(\ell)} }} -1 \right )N_0 B_v {\bar T}_{{\rm{UE}},v}^{(\ell)}}{d_v^{-\alpha}G_{\mathrm{UE},v}h_{v}^{(\ell)}},  & \text{ if } v \in \mathcal{N}.
\end{cases}
\end{equation}
\end{theorem}

\begin{proof}
    The proof is straightforward and is omitted due to space limitations.
\end{proof}

\begin{remark}
Consider the case $K=1$, where each token activates only one expert at each layer. In this case, the feasible set $\Omega^{(\ell)}$ reduces to a set of individual edge nodes whose hosted experts satisfy the latency and accuracy constraints. Consequently, $\mathcal{P}1$ admits a closed-form solution: the optimal expert is the one associated with the minimum per-expert energy cost $f_v^{(\ell)}$ among all feasible candidates. That is, the optimal decision is obtained by selecting
$v^\ast = \arg\min_{v \in \Omega^{(\ell)}} f_v^{(\ell)}$,
and setting $x_{v^\ast}^{(\ell)}=1$ and $x_v^{(\ell)}=0$ for all $v\neq v^\ast$.
\end{remark}

\subsection{Prefilling Phase with Multiple Tokens}
In this subsection, we consider the prefill phase where multiple input tokens are processed simultaneously. Unlike the decoding phase where only a single token is processed, tokens compete for the resources of the same helper. Therefore, expert selection decisions must be jointly optimized across tokens.

Each token may be offloaded to one or more helpers, which is governed by $K$. When $K=1$, each token is assigned to exactly one helper, and the resulting optimization problem reduces to a one-to-many assignment. 
When $K>1$, each token selects a subset of helpers, and the feasible decisions are defined over combinations of helpers rather than individual nodes. This introduces additional combinatorial coupling among helpers within each token’s decision.
Consequently, the structure of the optimization problem with $K>1$ is fundamentally different from that with $K=1$, and the algorithm designed for the $K=1$ case cannot be directly applied. 
We therefore distinguish these two settings and analyze them separately: we first consider the case with $K=1$, and then study the case with $K \geq 1$.

\subsubsection{When $K=1$}
Recall that $E_n^{(\ell)}$ is a non-increasing function of $T_{{\mathrm{UE}},n}^{(\ell)}$. Therefore, the uplink transmission time $T_{{\mathrm{UE}},n}^{(\ell)}$ should be maximized.
When $I_n^{(\ell)} > 0$, the optimal uplink transmission latency is given by
$T_{{\mathrm{UE}},n}^{(\ell) \ast} = T - D_n^{(\ell)} \left( \frac{\phi}{C_n} + \frac{b}{R_{n,{\mathrm{UE}}}^{(\ell)}}  \right) $. Otherwise, $T_{{\mathrm{UE}},n}^{(\ell) \ast} = 0$.
Therefore,  the optimal uplink energy consumption of transmitting hidden states to helper $n$ can be given by $E_n^{(\ell) \ast}  	 = 
	\frac{\left ( 2^{\frac{D_n^{(\ell)} b}{B_n T_{\mathrm{UE},n}^{(\ell) \ast} }} -1 \right )N_0 B_n T_{\mathrm{UE},n}^{(\ell) \ast}}{d_n^{-\alpha}G_{\mathrm{UE},n}h_{n}^{(\ell)}}$.
In particular, when $D_n^{(\ell)} = 0$, $E_n^{(\ell)\ast} = 0$ holds.

With the above optimal transmission strategy for any given load, we now analyze the properties of the energy cost function to facilitate the token-to-helper assignment.

\begin{proposition}\label{prop:discrete-convex}
$E_{\rm{UE}}^{(\ell)}$ is a monotonically increasing and discretely linear function of $D_{\rm{UE}}^{(\ell)}$, and
    $E_n^{(\ell)} $ is a monotonically increasing and discretely convex function of $D_n^{(\ell)}$.
\end{proposition}

\begin{proof}
    Please see Appendix \ref{proof:prop-discrete-convex}.
\end{proof}

From Proposition~\ref{prop:discrete-convex}, the following analysis focuses on $E_n^{(\ell)} $, which exhibits a nonuniform marginal-cost structure.
The term $E_{\rm{UE}}^{(\ell)}$ corresponds to the constant marginal-cost case and can be treated as a special case.
Define the marginal cost of assigning the $d$-th token to helper $n$ as $\Delta E_n^{(\ell)}(d)\triangleq E_n^{(\ell)}(d)-E_n^{(\ell)}(d-1)$. Since $E_n^{(\ell)}$ is discretely convex in $D_n^{(\ell)}$, the marginal cost sequence $\{\Delta E_n^{(\ell)}(d)\}$ is nondecreasing in $d$.
Accordingly, the energy cost of transmitting to helper $n$ can be expressed as an accumulated sum of marginal costs, i.e., $E_n^{(\ell)}(D_n^{(\ell)})=\sum_{d=1}^{D_n^{(\ell)}}\Delta E_n^{(\ell)}(d)$. This decomposition admits an exact slot-based interpretation: assigning $D_n^{(\ell)}$ tokens to helper $n$ is equivalent to occupying the first $D_n^{(\ell)}$ unit-capacity slots, where occupying the $d$-th slot incurs a cost $\Delta E_n^{(\ell)}(d)$.
Due to the monotonicity of the marginal costs, any feasible assignment that occupies a higher-index slot while leaving a lower-index slot unused can be transformed into another feasible assignment with no higher cost by exchanging these slots. Repeating this exchange yields an optimal assignment whose occupied slots form a prefix $\left \{ 1,\ldots,D_n^{(\ell)} \right \}$ for each helper. Hence, the slot representation constitutes an exact reformulation of the original discrete-convex objective.

When $K=1$, each feasible combination in $\Omega_m^{(\ell)}$ contains a single helper. Under the exact slot-based reformulation, the token–expert matching problem can be cast as a linear min-cost flow problem and solved optimally using a successive shortest augmenting path (SSAP) algorithm~\cite{ahuja1993network}.  Algorithm~\ref{alg:dc_matching} implements this procedure without constructing all slot nodes. 
Specifically, we construct a residual network consisting of a source node $s$, token nodes $\{u_m\}_{m\in\mathcal{M}}$, expert nodes $\{q_v\}_{v\in\mathcal{V}}$, and a sink node $t$. Edges $s \to u_m$ represent the initiation of token assignments, while edges $q_v \to t$ capture the marginal energy cost incurred when assigning additional tokens to node $q_v$. All arcs have unit capacity unless otherwise specified.
Under the discrete-convex objective, any load-preserving reassignment corresponds to a zero-cost cycle in the residual network. Augmenting flow along such cycles enables the SSAP algorithm to revise earlier assignments without increasing the total cost. 

The optimality of Algorithm~\ref{alg:dc_matching} follows from standard results on minimum-cost flow, where the SSAP algorithm is guaranteed to find a globally optimal solution~\cite{ahuja1993network}.
Since the graph may contain negative weights, the Bellman-Ford algorithm can be adopted to find the shortest path. Following standard complexity results for the SSAP algorithm~\cite{ahuja1993network},  the computational complexity of Algorithm~\ref{alg:dc_matching} is $\mathcal{O} \left ( M \left ( M+N \right ) \left (M+N+ \sum_{m \in \mathcal{M} } \left | \Omega_m^{(\ell)} \right | \right )    \right ) $.

\begin{algorithm}[t]
\caption{Successive shortest augmenting path-based token-expert matching when $K=1$}
\label{alg:dc_matching}
\begin{algorithmic}[1]
\State \textbf{Input:} Token set $\mathcal{M}$, edge node set $\mathcal{V}$, computing
capacity limits $\{D_{v,\max}^{(\ell)}\}_{v\in\mathcal{V}}$;
feasible sets $\{\Omega_m^{(\ell)}\}_{m \in \mathcal{M}}$;
marginal energy costs $\{\Delta E_v(d)\}_{v\in\mathcal{V}, d\leq D_{v,\max}^{(\ell)}}$.

\State Initialize $D_v^{(\ell)} \gets 0$, $\forall v\in\mathcal{V}$, $x_{m,v}^{(\ell)} \gets 0$ for all admissible $(m,v)$, and total required flow $F \gets |\mathcal{M}|$.

\While{$F>0$}
    \State Construct a residual network $\mathcal{G}_{\mathrm{res}}$ with nodes $\{s\}\cup\{u_m\}_{m\in\mathcal{M}}\cup\{q_v\}_{v\in\mathcal{V}}\cup\{t\}$.

    \For{each token $m\in\mathcal{M}$}
        \If{$\sum_{v\in \Omega_m^{(\ell)}} x_{m,v}^{(\ell)} = 0$}
            \State Add  arc $(s\to u_m)$ with cost $0$.
        \Else
            \State Add arc $(u_m\to s)$ with cost $0$.
        \EndIf
    \EndFor

    \For{each admissible pair $(m,v)$ with $v\in\Omega_m^{(\ell)}$}
        \If{$x_{m,v}^{(\ell)}=0$}
            \State Add arc $(u_m\to q_v)$ with cost $0$.
        \Else
            \State Add arc $(q_v\to u_m)$ with cost $0$.
        \EndIf
    \EndFor

    \For{each edge node $v\in\mathcal{V}$}
        \If{$D_v^{(\ell)} < D_{v,\max}^{(\ell)}$}
            \State Add arc $(q_v\to t)$ with cost $\Delta E_v\!\left(D_v^{(\ell)}+1\right)$.
        \EndIf
        \If{$D_v^{(\ell)} > 0$}
            \State Add arc $(t\to q_v)$ with cost $-\Delta E_v\!\left(D_v^{(\ell)}\right)$.
        \EndIf
    \EndFor

    \State Find a shortest path $P$ from $s$ to $t$ in $\mathcal{G}_{\mathrm{res}}$ using Bellman-Ford algorithm.
 If no such path exists, terminate and declare the problem infeasible. Otherwise, augment one unit of flow along $P$.

    \For{each arc $(a\to b)$ on path $P$}
        \If{$(a\to b)=(u_m\to q_v)$} \State $x_{m,v}^{(\ell)} \gets 1$, $D_v^{(\ell)} \gets D_v^{(\ell)} +1$. \ElsIf{$(a\to b)=(q_v\to u_m)$} \State $x_{m,v}^{(\ell)} \gets 0$, $D_v^{(\ell)} \gets D_v^{(\ell)} -1$. \EndIf
    \EndFor

    \State $F \gets F-1$.
\EndWhile

\State \Return $\left \{ x_{m,v}^{(\ell)} \right \}$.
\end{algorithmic}
\end{algorithm}

\subsubsection{When $K \geq 1$}

When each token can be assigned to multiple edge nodes, i.e., $K>1$, the structure of $\mathcal{P}1$ fundamentally changes. 
The feasible decision for each token is no longer a single edge node but a combination of edge nodes, which introduces combinatorial coupling across selection variables. As a result, $\mathcal{P}1$ becomes NP-hard. Nevertheless, two structural properties enable efficient algorithm design. First, the marginal cost formulation derived for the $K=1$ case still applies. Second, the system evolves sequentially across tokens and can be characterized by the helper load vector, which naturally leads to a load-based system state.

These structural properties motivate a DP formulation. By representing the system state using the helper load vector, the problem can be modeled as a sequential decision process across tokens. Accordingly, we develop a DP-based algorithm that yields the globally optimal solution for moderate problem sizes, as shown in Algorithm \ref{alg:dp_opt}. 
Specifically, each token $m \in \mathcal{M}$ selects exactly one feasible edge node combination $J_m^{(\ell)} \in \Omega_m^{(\ell)}$, where each combination satisfies $1 \le |J_m^{(\ell)}| \le K$.
For token $m$, let $\mathrm{DP}_m^{(\ell)}(\mathbf{D})$ denote the minimum total energy incurred by the first $m$ tokens, given the current load vector over all edge nodes $\mathbf{D} = \left \{ D_v^{(\ell)} \right \}_{v \in \mathcal{V}} $. When token $m$ selects a combination $J_m^{(\ell)}$, the system transitions to a successor state $\tilde{\mathbf{D}} = \left \{\tilde{D}_v^{(\ell)} \right \}_{v \in \mathcal{V}} $, where the load of each edge node $v$ is incremented by one. By iterating over tokens and feasible combinations, the DP algorithm explores all admissible load evolutions.

The optimality of Algorithm \ref{alg:dp_opt} follows from Bellman’s principle of optimality, as the problem exhibits an optimal substructure and the DP formulation exhaustively explores all feasible state transitions~\cite{bellman1957dynamic}.
Following standard DP procedure~\cite{bellman1957dynamic}, it is easy to show that the upper bound on the computational complexity of Algorithm \ref{alg:dp_opt} is $
\mathcal{O}\left(
M \Omega_{\max}  K \prod_{n}(D_{n,\max}^{(\ell)}+1)
\right)$, where $\Omega_{\max} = \max_{m \in \mathcal{M}} \left | \Omega_m^{(\ell)} \right | $. 
Note that Algorithm \ref{alg:dp_opt} also applies to the case $K=1$. However, for $K=1$, the problem admits a simpler structure and can be solved more efficiently by Algorithm \ref{alg:dc_matching}.

 \begin{algorithm}[t]
\caption{DP-based token-expert matching when $K \geq 1$}
\label{alg:dp_opt}
\begin{algorithmic}[1]

\State \textbf{Input:}
Token set $\mathcal{M}$, edge node set $\mathcal{V}$, computing
capacity limits $\{D_{v,\max}^{(\ell)}\}_{v\in\mathcal{V}}$;
feasible sets $\{\Omega_m^{(\ell)}\}_{m \in \mathcal{M}}$;
marginal energy costs $\{\Delta E_v(d)\}_{v\in\mathcal{V}, d\leq D_{v,\max}^{(\ell)}}$.

\State Initialize DP table:
$\mathrm{DP}_0^{(\ell)}(\mathbf{D}) \gets +\infty$ for all states $\mathbf{D}$,
and $\mathrm{DP}_0^{(\ell)}(\mathbf{0}) \gets 0$,
where $\mathbf{D}=\left\{ D_v^{(\ell)}\right\}_{v \in \mathcal{V}}$.

\State Initialize backtracking pointers:
$\mathrm{ParState}_0^{(\ell)}(\mathbf{0}) \gets \varnothing$,
$\mathrm{ParComb}_0^{(\ell)}(\mathbf{0}) \gets \varnothing$.

\For{$m=1$ to $M$}
    \State Initialize $\mathrm{DP}_m^{(\ell)}(\mathbf{D}) \gets +\infty$ for all states $\mathbf{D}$.
    
    \For{each state $\mathbf{D}  $ with $\mathrm{DP}_{m-1}^{(\ell)}(\mathbf{D})<+\infty$}
        \For{each combination $J \in\Omega_m^{(\ell)}$}
            \If{$D_v^{(\ell)}+1\le D_{v,\max}^{(\ell)}$ for all $v \in J$}
                \For{each edge node $v\in\mathcal{V}$}
                    \If{$v\in J$}
                        \State $\tilde{D}_v^{(\ell)} \gets D_v^{(\ell)}+1$
                    \Else
                        \State $\tilde{D}_v^{(\ell)} \gets D_v^{(\ell)}$
                    \EndIf
                \EndFor
                
                \State 
                $\Delta \gets \sum \limits_{v\in J}\Delta E_v\left(D_v^{(\ell)}+1\right) $.
                
                \If{$\mathrm{DP}_{m-1}^{(\ell)}(\mathbf{D})+\Delta
                < \mathrm{DP}_m^{(\ell)}(\mathbf{\tilde{D}})$}
                    \State $\mathrm{DP}_m^{(\ell)}(\mathbf{\tilde{D}})
                    \gets \mathrm{DP}_{m-1}^{(\ell)}(\mathbf{D})+\Delta$;
$\mathrm{ParState}_m^{(\ell)}(\mathbf{\tilde{D}}) \gets \mathbf{D}$;
 $\mathrm{ParComb}_m^{(\ell)}(\mathbf{\tilde{D}}) \gets J$.
                \EndIf
            \EndIf
        \EndFor
    \EndFor
\EndFor

\State Select terminal state: $\mathbf{D}^{\ast}
\gets \arg\min_{\mathbf{D}} \mathrm{DP}_M^{(\ell)}(\mathbf{D})$.
 Backtracking to recover $\{J_m^{(\ell)\star}\}$: $\mathbf{D}_{\mathrm{cur}} \gets \mathbf{D}^{\ast}$.
\For{$m=M$ down to $1$}
    \State $J_m^{(\ell)\star}
    \gets \mathrm{ParComb}_m^{(\ell)}(\mathbf{D}_{\mathrm{cur}})$.
 $\mathbf{D}_{\mathrm{cur}}
    \gets \mathrm{ParState}_m^{(\ell)}(\mathbf{D}_{\mathrm{cur}})$.
\EndFor

\State \Return Optimal expert subset selection $\{J_m^{(\ell)\star}\}$.

\end{algorithmic}
\end{algorithm}

 \subsubsection{Design insights}
The following remarks provide insight into the structural reasons behind the optimality of the proposed algorithms for arbitrary $K$ and characterize the optimal transmission strategy.

\begin{remark}[Marginal-cost curvature over feasible edge loads]
    The optimal token–expert assignment is governed entirely by the growth rate of the marginal energy cost with respect to edge load.
    For helper nodes, the discretely convex uplink energy function induces a load-dependent marginal cost whose curvature determines the assignment behavior.
When the marginal cost increases steeply due to limited computing capability, unfavorable user–helper channel conditions, or tight transmission time constraints, assigning many tokens to the same helper becomes energy-inefficient, and tokens are preferentially assigned to other helpers with lower marginal costs.
Conversely, when the marginal cost grows slowly over the feasible load region, additional tokens tend to be assigned to that helper until its marginal cost becomes comparable to others.
The user (local computing) corresponds to a special linear case with constant marginal cost. 
\end{remark}

\begin{remark}[Cost symmetry and non-uniqueness]
The system cost depends only on the edge node's load vector $D_v^{(\ell)}$ and is invariant to the identities of individual tokens. This load-based symmetry implies that tokens with nested feasible edge node sets can be reassigned or exchanged across edge nodes without violating optimality, as long as the resulting load vector remains unchanged. 
\end{remark}

\begin{remark}[Uniform bit transmission under slow fading]\label{uniform}
Under slow fading, the uplink transmission energy is a convex function of the number of transmitted bits over the layer duration.
Therefore, for a given uplink transmission time budget, evenly distributing the required bits across time slots minimizes the total energy consumption.
\end{remark}

\section{Joint Expert Selection and adaptive transmission under fast fading}\label{sec:fast_fading}
In contrast to the slow-fading case discussed in Section \ref{sec:slow_fading}, fast fading introduces time-varying channel conditions within each MoE layer. Under such dynamics, uniformly distributing transmitted bits (in Remark \ref{uniform}) across time slots is no longer energy-optimal, and exploiting channel variations through adaptive slot-level bit allocation becomes necessary. In this section, we study the joint optimization of expert selection and per-slot bit allocation under fast fading to minimize expected energy consumption.

\subsection{Problem Reformulation}
To model fast fading within each MoE layer, we partition the uplink transmission phase into multiple time slots of equal duration $\tau$. 
The channel gain is assumed to be quasi-static, i.e., being constant within each slot but varying independently across slots.
Accordingly, the uplink transmission energy incurred by helper $n$ in time slot $t$ of MoE layer $\ell$ is given by
\begin{equation}
	E_{n,t}^{(\ell)}
	=
	\frac{\left ( 2^{\frac{\gamma_{n,t}^{(\ell)}}{B_n \tau }} -1 \right )N_0 B_n \tau }{d_n^{-\alpha}G_{\mathrm{UE},n} h_{n,t}^{(\ell)}},
\end{equation}
where $\gamma_{n,t}^{(\ell)}$ is the number of bits transmitted from the user to helper $n$ in slot $t$ of layer $\ell$, and $h_{n,t}^{(\ell)}$ is the corresponding fast-fading channel gain.

We now formulate the joint optimization of expert selection $\mathbf{X}^{(\ell)}$ and slot-level bit allocation $\bm{\gamma}^{(\ell)}$ under fast fading. Specifically, for each MoE layer $\ell$, the problem is given by
\begin{subequations}
	\begin{equation}\label{equ:obj_P2}
		\mathcal{P}2: \quad		\min_{\mathbf{X}^{(\ell)}, \bm{\gamma}^{(\ell)}} \;  E_{\rm{UE}}^{(\ell)}+
        \mathbb{E}_{h} \left[ \sum_{n \in \mathcal{N}}\sum_{t=1}^{Q_n^{(\ell)}}  E_{n,t}^{(\ell)}   \right]
	\end{equation}
	\begin{equation}
		{\rm{s.t.}} \quad  (\mathrm{\ref{constraint:P1_latency} }) - (\mathrm{\ref{constraint:P1_binary}}),
	\end{equation}
\begin{equation}\label{constraint:P3layer-sumbit}
    \sum_{t=1}^{Q_n^{(\ell)}} \gamma_{n,t}^{(\ell)} = b D_n^{(\ell)}, \forall n \in \mathcal{N}, \ell \in \mathcal{L},
    \end{equation}
\end{subequations}
where $Q_n^{(\ell)} = \left \lceil \frac{T_{\mathrm{UE},n}^{(\ell)}}{\tau} \right \rceil  $ denotes the number of uplink transmission slots allocated to helper $n$ in layer $\ell$.

Due to the separable sum-bit constraint across helpers, the inner optimization over $\bm{\gamma}^{(\ell)}$ can be decomposed. Then, the total energy consumption of helpers in (\ref{equ:obj_P2}) can be equivalently rewritten as 
\[\min_{\mathbf{X}^{(\ell)}, \bm{\gamma}^{(\ell)}} 
        \mathbb{E}_{h} \left[ \sum_{n \in \mathcal{N}}\sum_{t=1}^{Q_n^{(\ell)}}  E_{n,t}^{(\ell)}   \right]
=
\min_{\mathbf X^{(\ell)}}\sum_{n\in\mathcal N}
E_{n}^{(\ell)\ast} \left(D_n^{(\ell)}\right),\]
where $E_{n}^{(\ell)\ast}\left ( D_n^{(\ell)} \right )  \;\triangleq\; \min_{\pi_n} \mathbb E\!\left[ \sum_{t=1}^{Q_n^{(\ell)}} E_{n,t}^{(\ell)}\!\bigl(\gamma_{n,t}^{(\ell)},h_{n,t}^{(\ell)}\bigr) \right]$ is the minimum expected uplink energy required to transmit $bD_n^{(\ell)}$ bits to helper $n$ under fast fading, and $\pi_n$ represents a causal bit-allocation policy that adapts $\gamma_{n,t}^{(\ell)}$ to the instantaneous fading $h_{n,t}^{(\ell)}$. 

\subsection{Expert Selection Scheme}
Although exact, the value function $E^{(\ell)\ast}_{n}\left ( D_n^{(\ell)} \right ) $ generally admits no closed-form expression and is difficult to evaluate. To enable tractable expert selection, we introduce a deterministic surrogate by replacing the random fading term with its expectation. 
Specifically, we define $\tilde E_n^{(\ell)}(D) \triangleq \mathbb E\!\left[\frac{1}{h}\right] \frac{\left( 2^{\frac{bD}{B_nT_{\mathrm{UE},n}^{(\ell)}}}-1 \right) N_0 B_nT_{\mathrm{UE},n}^{(\ell)}} {d_n^{-\alpha}G_{\mathrm{UE},n}}$.
The surrogate $\tilde E_n^{(\ell)}(D)$ corresponds to minimizing the expected energy over a restricted class of non-adaptive bit-allocation policies that do not exploit instantaneous fading variations. Since $E_{n}^{(\ell)\ast}(D)$ minimizes over all causal policies, the surrogate provides an upper bound: $E_{n}^{(\ell)\ast}(D)\ \le\ \tilde E_n^{(\ell)}(D), \forall D$.
Using this surrogate, the outer problem reduces to $\min \limits_{\mathbf X} \; E_{\rm{UE}}^{(\ell)}+ \sum_{n\in\mathcal N} \tilde E_n^{(\ell)}\left ( D_n^{(\ell)}  \right )$, which has the same structural form as the slow-fading case and can therefore be solved using the proposed algorithms in Section \ref{sec:slow_fading}.

\subsection{Adaptive Transmission Strategy}
After obtaining the optimal helper loads $\left \{ D_n^{(\ell) \ast} \right \} $ from the outer problem, the optimal slot-level bit allocation $\left \{\gamma_{n,t}^{(\ell)} \right \}$ is recovered by solving the inner problem for helper $n$, which is given by
\begin{equation}
		\mathcal{P}3_n: \quad		\min_{\bm{\gamma}^{(\ell)}} \;   \mathbb E_{h} \left[ \sum_{t=1}^{Q_n^{(\ell)}} E_{n,t}^{(\ell)}\!\left(\gamma_{n,t}^{(\ell)}, h_{n,t}^{(\ell)}\right) \right] 
	\end{equation}
    with the constraint $\sum_{t=1}^{Q_n^{(\ell)}} \gamma_{n,t}^{(\ell)} = b  D_n^{(\ell)\ast}$.

Let $\beta_{n,t}^{(\ell)}$ denote the remaining number of bits to be transmitted to helper $n$ at the beginning of time slot $t$ during MoE layer $\ell$.
    Using the approach of DP, the objective function of $\mathcal{P}3_n$ can be rewritten as 
    \begin{equation}
    \begin{split}
    &\mathcal{P}4_n:	U_t\left(  \beta_{n,t}^{(\ell)},  h_{n,t}^{(\ell)}        \right) \\
    & = 
	\begin{cases}
		\min \limits_{\gamma_{n,t}^{(\ell)} \leq \beta_{n,t}^{(\ell)}} E_{n,t}^{(\ell)}+	{\bar{U}}_t\left(  \beta_{n,t}^{(\ell)} -\gamma_{n,t}^{(\ell)}   \right),  & \text{ if } 1\leq t <Q_n^{(\ell)}, \\
		\frac{\left ( 2^{a_n  \beta_{n,t}^{(\ell)}} -1 \right )c_n }{h_{n,t}^{(\ell)}}, & \text{ if } t=Q_n^{(\ell)},
	\end{cases}
    \end{split}
\end{equation}
where $a_{n} = \frac{1}{B_n \tau} $, $c_n = \frac{N_0 B_n \tau }{d_n^{-\alpha}G_{\mathrm{UE},n}}$.
The cost-to-go function
$ {\bar{U}}_t\left( \beta_{n,t}^{(\ell)} \right)  = \mathbb{E}_h \left [ U_t\left(  \beta_{n,t}^{(\ell)} ,  h_{n,t}^{(\ell)}      \right) \right ]  $ denotes the expected energy consumption for transmitting $\beta_{n,t}^{(\ell)}$ bits over the remaining slots.

\begin{theorem}[Optimal slot-level bit allocation under fast fading]\label{theorem:fastfading}
    Under fast-fading channels, considering a fixed helper $n$ and a given load $D_n^{(\ell)}$ at MoE layer $\ell$, the optimal bit allocation in slot $t$ is given by
    \begin{equation}
    \begin{split}
         \gamma_{n,t}^{(\ell)\ast} & = \frac{1}{\left ( Q_n^{(\ell)}-t+1 \right ) a_n} \\
 & \cdot \left[ a_n\beta_{n,t}^{(\ell)} +\sum_{j=0}^{Q_n^{(\ell)}-t-1}\log_2\left( h_{n,t}^{(\ell)} \mathbb{E}\!\left[\frac{1}{h_{n,Q-j}^{(\ell)}}\right] \right) \right].
    \end{split}
    \end{equation}

    The resulting minimum expected uplink energy consumption is given by
\begin{equation}
    \begin{split}
\tilde{E}_n^{\ast}\left( D_n^{(\ell)} \right) & = c_n\left[ Q_n^{(\ell)} 2^{\frac{a_n D_n^{(\ell)}}{Q_n^{(\ell)}}}\left(\frac{1}{h_{n,1}^{(\ell)}} \prod_{t=2}^{Q_n^{(\ell)}}\mathbb{E}\left [ \frac{1}{ h_{n,t}^{(\ell)}  } \right ]\right) \right. \\
        & \left. - \left(\frac{1}{h_{n,1}^{(\ell)}}+\sum_{t=2}^{Q_n^{(\ell)}}\mathbb{E}\left [ \frac{1}{ h_{n,t}^{(\ell)}  } \right ]	\right)
		\right].
    \end{split}
\end{equation}
\end{theorem}

\begin{proof}
    Please see Appendix \ref{proof:theorem-fastfading}.
\end{proof}

\section{Experimental Results}\label{sec:experiment}
\subsection{Experimental Setup}
In our simulations, we consider three types of MoE models. The first is Switch-base-8~\cite{fedus2022switch}, a Switch Transformer-based MoE model with 8 experts per MoE layer, which is tested on the Extreme Summarization (XSum) dataset~\cite{narayan2018don}. Under the conventional score-based Top-$K$ routing mechanism, each token is routed to a single expert at each layer (i.e., $K=1$).
The second is Mixtral-8$\times$7B~\cite{jiang2024mixtral}, an MoE model with 8 experts per MoE layer, which is tested on the CommonsenseQA dataset~\cite{talmor2019commonsenseqa}. Under the Top-$K$ routing mechanism, each token is routed to two experts at each layer (i.e., $K=2$).
The third is GPT-OSS-20B~\cite{agarwal2025gpt}, an MoE model with 32 experts per MoE layer, which is evaluated on the GSM8K dataset~\cite{cobbe2021trainingverifierssolvemath}. Under its Top-$K$ routing mechanism, each token is routed to four experts at each layer (i.e., $K=4$).
The expert loading latency and computation latency are measured on an NVIDIA GeForce RTX 4090 GPU. To improve the realism of the wireless evaluation, the user-helper distances are generated from the Microsoft GeoLife GPS trajectory dataset~\cite{zheng2011geolife}, where helpers are fixed and the user location follows real GPS traces mapped to a service area with a maximum distance of 150 m.
The wireless parameters are set as follows~\cite{11395617}: The transmit power of helpers is 38 dBm. The path-loss factor is $\alpha = 4$, the antenna-related factor is 1, and the noise power spectral density is -174 dBm/Hz. The small-scale channel fading $h$ is modeled as a Gamma-distributed fading variable with unit mean, where the shape parameter is set to 2~\cite{7880703}. The hyperparameters introduced in Assumptions \ref{assumption:bounded_output} and \ref{assumption:Lip_expert} can be empirically estimated using finite-difference tests following the lightweight approach in~\cite{11314800,wei2025optimizing}.

To demonstrate the effectiveness of our proposed SiftMoE during WIDE inference, we compare SiftMoE with two benchmark schemes:
\begin{itemize}
    \item \textbf{Ideal Top-K}: This benchmark follows the conventional score-based Top-$K$ expert routing, where the outputs of all selected experts are always successfully received by the user within the given latency constraint. It therefore serves as an upper bound on the achievable inference accuracy.

    \item  \textbf{(Practical) Top-K}: This scheme also adopts the conventional score-based Top-$K$ routing but accounts for realistic wireless channel conditions. In this case, some expert outputs may fail to be delivered to the user within the latency budget, leading to degraded inference performance.

 \item \textbf{WDMoE}~\cite{11083676}: This benchmark adopts a gating-score-based adaptive expert activation strategy. It sorts experts according to their gating scores and sequentially activates them until the cumulative gating score exceeds a predefined threshold, thereby dynamically determining the number of activated experts for each token.

    \item \textbf{AdaptMoE}~\cite{huang-etal-2024-harder}: This benchmark adopts a wireless-aware adaptive expert selection strategy for WIDE inference. It jointly considers expert gating weights and the latency of the corresponding experts to dynamically drop low-weight experts, while further coupling expert selection with wireless bandwidth allocation.

\end{itemize}

\subsection{Effectiveness of Proposed Offline-online Combined Estimation Method}

To validate the practical effectiveness of the proposed offline-online combined estimation method, we compare three quantities in Fig. \ref{fig:layer_deviation}: the measured layer-output deviation, the Proposition-1 error bound, and the proposed offline-online combined estimation. Specifically, the measured layer-output deviation refers to the actual average output deviation at each MoE layer caused by SiftMoE. The Proposition-1 error bound is computed using the right-hand side of the inequality in (\ref{equ:prop1_layerwise}), where the input-dependent expert-output mismatch is obtained by executing both the original and substitute experts. The proposed offline-online combined estimation corresponds to (\ref{equ:offline_online_layerwise}), where the expert scores are obtained online from the gating network, while the expert mismatches are estimated offline using a calibration dataset.

Fig. \ref{fig:prop1_deviation_vs_tau_switch} shows the results for the Switch-base-8 model. The analytical upper bound overlaps with the measured output deviation over the whole range of maximum tolerable errors. This is because only one expert is involved under Top-1 routing, and the triangle inequality used in the proof of Proposition \ref{prop:layer-expect-perturbation} becomes tight. Meanwhile, the proposed offline-online combined estimation closely follows the measured deviation with only a small conservative gap, indicating that the offline-calibrated mismatch can accurately approximate the practical layer-wise deviation.

Fig. \ref{fig:prop1_deviation_vs_tau_mixtral} presents the results for the Mixtral-8$\times$7B model. In this case, multiple activated experts jointly contribute to the layer output, making the deviation characterization more challenging. Nevertheless, the Proposition-1 error bound consistently remains above the measured output deviation and stays close to it, verifying both the validity and tightness of the proposed analytical bound. The proposed offline-online combined estimation also closely tracks the measured deviation, demonstrating that the calibrated expert-pair mismatch provides an accurate and lightweight surrogate for online expert selection. A similar trend can be observed for the GPT-OSS-20B model in Fig. \ref{fig:prop1_deviation_vs_tau_gpt}, further confirming that the proposed bound and estimation method remain effective when more experts are activated simultaneously.

Overall, these results show that the proposed offline-online combined estimation provides an accurate and efficient way to enforce the layer-wise deviation constraint in SiftMoE. The computationally expensive estimation of expert-pair mismatches is performed offline using a calibration dataset, while the online stage only combines the calibrated mismatches with the routing scores of the current Top-$K$ experts. The close match among the measured layer-output deviation, the Proposition-1 error bound, and the proposed offline-online combined estimation further indicates that the resulting expert-selection policy is not overly conservative, providing an effective basis for controlling the accuracy-energy tradeoff.

\begin{figure*}[t]
	\centering
	\subfigure[Switch-base-8]{\includegraphics[width =0.25\textwidth]{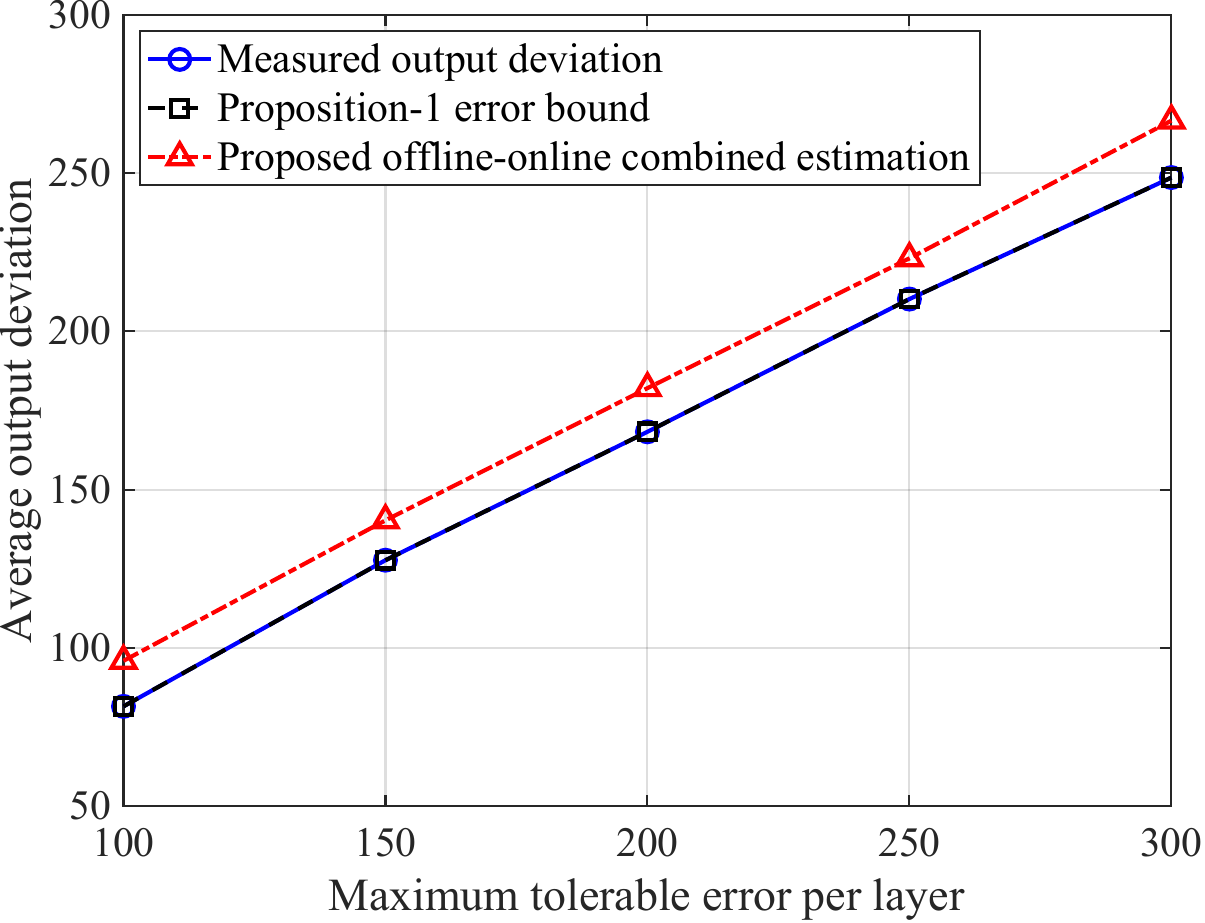}\label{fig:prop1_deviation_vs_tau_switch}}
    \quad 
	\subfigure[Mixtral-8$\times$7B]{\includegraphics[width =0.25\textwidth]{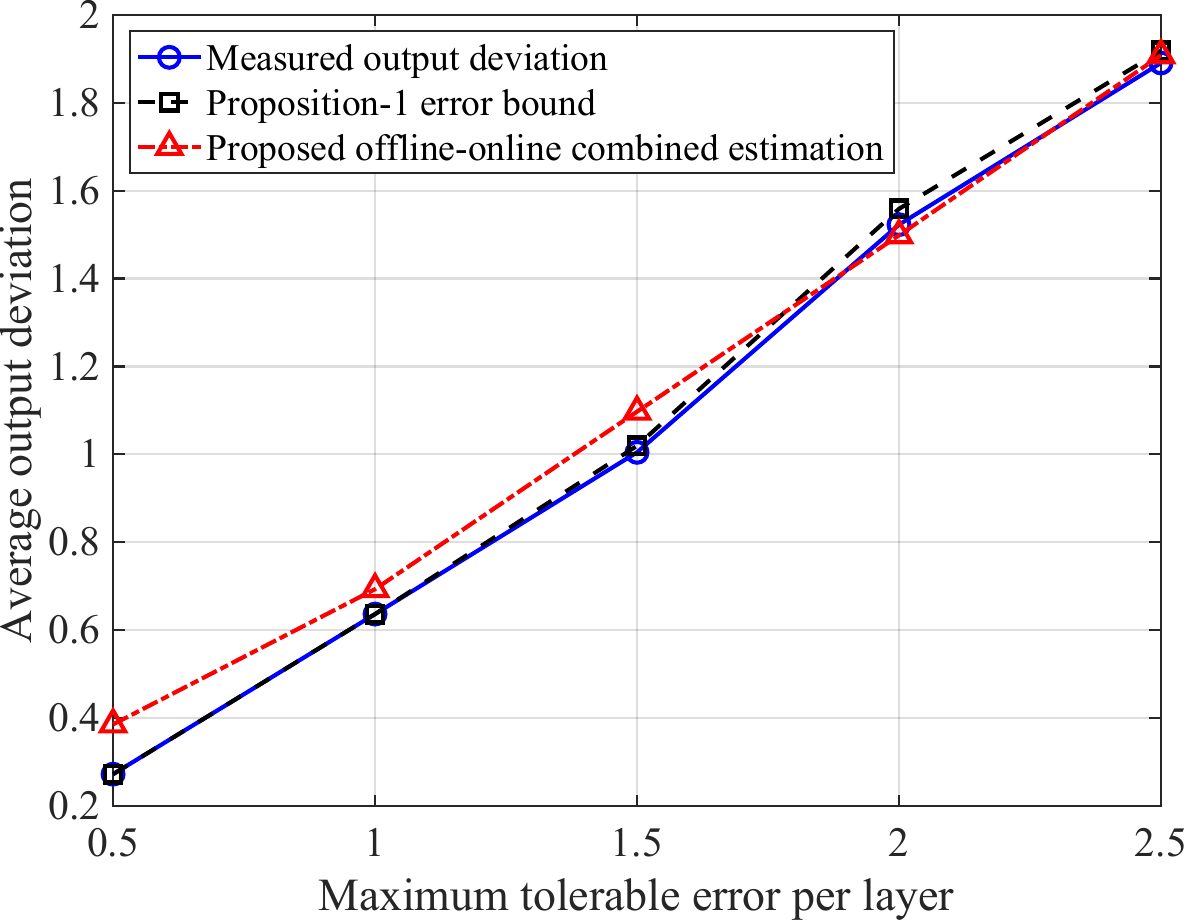}\label{fig:prop1_deviation_vs_tau_mixtral}}
    \quad 
    \subfigure[GPT-OSS-20B]{\includegraphics[width =0.25\textwidth]{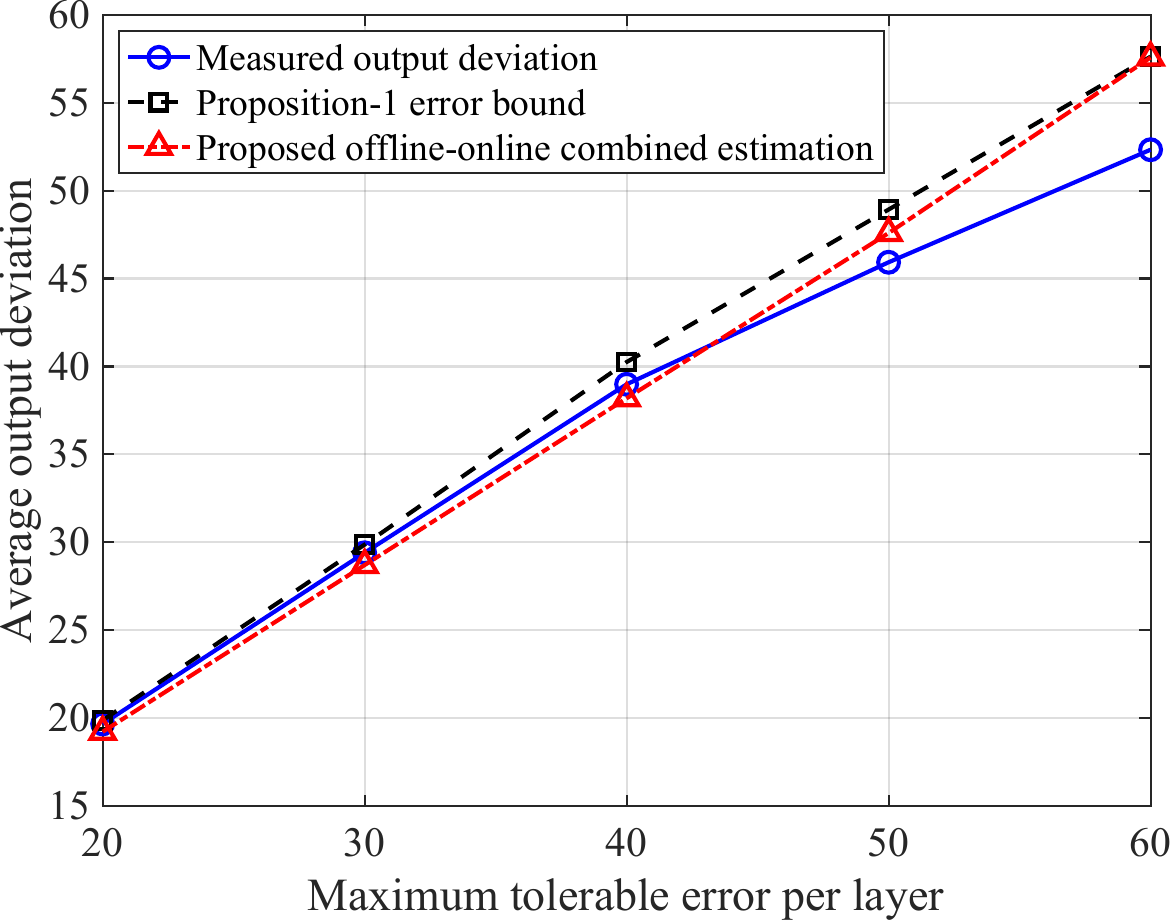}\label{fig:prop1_deviation_vs_tau_gpt}}
	\caption{Validation of the analytical upper bound and the proposed offline-online combined estimation.}
\label{fig:layer_deviation} 
\end{figure*}

\subsection{Effectiveness of SiftMoE and Accuracy-energy Tradeoff}\label{subsec:accuracy_energy}
To validate the effectiveness of expert replacement and skipping in the proposed SiftMoE, we compare the inference accuracy and energy consumption of SiftMoE with those of the benchmark schemes under slow-fading channels, where the impact of expert selection can be clearly observed.
Fig. \ref{fig:delta_switch_slow} illustrates the performance comparison on the Switch-base-8 model over XSum under different maximum tolerable errors per layer. As shown in Fig. \ref{fig:fig_AccDelta_switch_slow}, when the maximum tolerable error is small, the accuracy achieved by SiftMoE is very close to that of the ideal Top-1 benchmark, indicating that the proposed SiftMoE introduces only limited accuracy degradation. 
When the maximum tolerable error increases, the accuracy of SiftMoE gradually decreases due to accumulated layer-wise deviations. However, within a moderate range of tolerant errors, SiftMoE consistently outperforms the practical Top-1 scheme that accounts for realistic channel conditions. Meanwhile, Fig. \ref{fig:fig_EnergyDelta_switch_slow} shows that as the maximum tolerable error increases, the average energy consumption per token under SiftMoE decreases monotonically and remains significantly lower than that of Top-1 benchmark.
Fig. \ref{fig:fig_EnergyAccuracy_switch} further presents the resulting accuracy-energy tradeoff. Compared with Top-1 routing, WDMoE, and AdaptMoE, SiftMoE achieves a more favorable tradeoff between inference accuracy and communication energy consumption. This is because SiftMoE not only relies on expert scores or wireless latency for expert dropping, but also exploits functional expert similarity to enable accuracy-aware expert selection.

Fig. \ref{fig:delta_mixtral_slow} and Fig. \ref{fig:delta_gpt_slow} present the corresponding results for the Mixtral-8$\times$7B model on CommonsenseQA and GPT-OSS-20B on GSM8K, respectively. Similar performance trends can be observed in both cases: SiftMoE reduces the average energy consumption per token as the maximum tolerable error increases, while maintaining controllable accuracy degradation. Moreover, the accuracy-energy tradeoff curves show that SiftMoE consistently outperforms the corresponding Top-K, WDMoE, and AdaptMoE baselines. These results demonstrate that by appropriately tuning the maximum tolerable error, SiftMoE enables a flexible accuracy-energy tradeoff and remains effective across different MoE architectures and datasets.

These results demonstrate that by appropriately tuning the maximum tolerable error, SiftMoE enables a flexible accuracy–energy trade-off, where substantial energy savings can be achieved at the cost of controlled and predictable accuracy degradation.

\begin{figure*}[!t]
	\centering
	\subfigure[Accuracy.]{\includegraphics[width =0.25\textwidth]{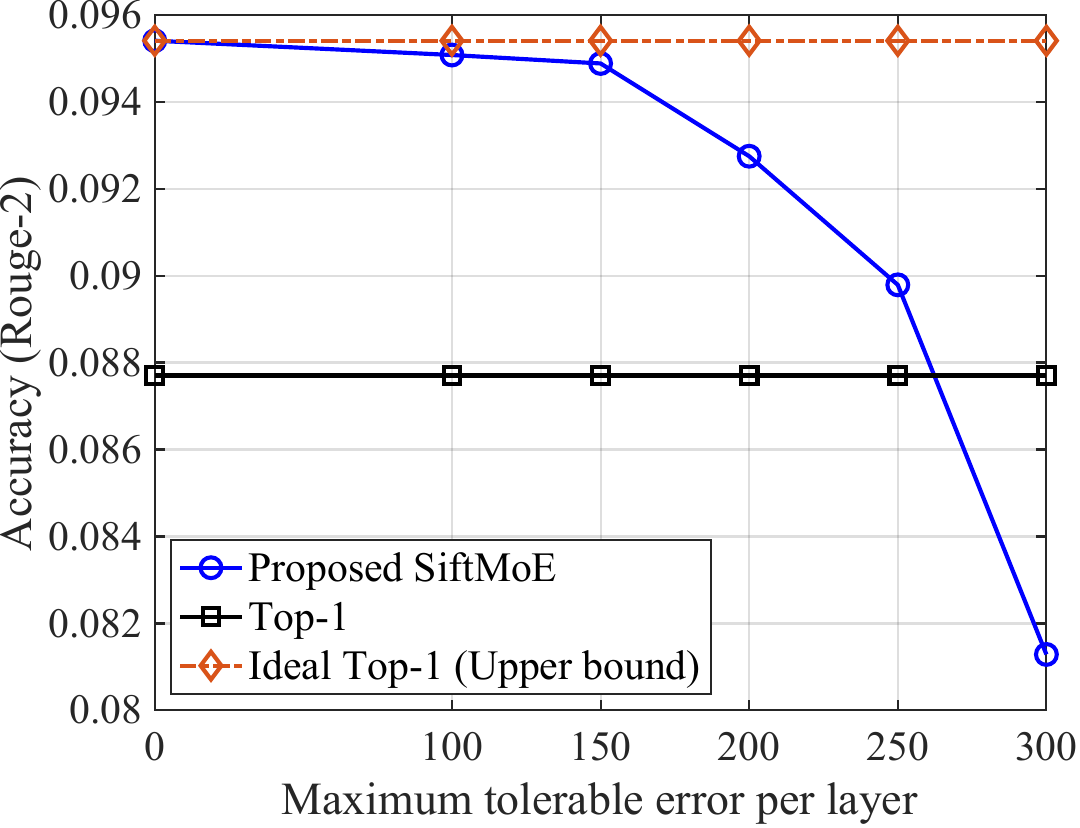}\label{fig:fig_AccDelta_switch_slow}}
    \quad 
	\subfigure[Energy consumption per token.]{\includegraphics[width =0.25\textwidth]{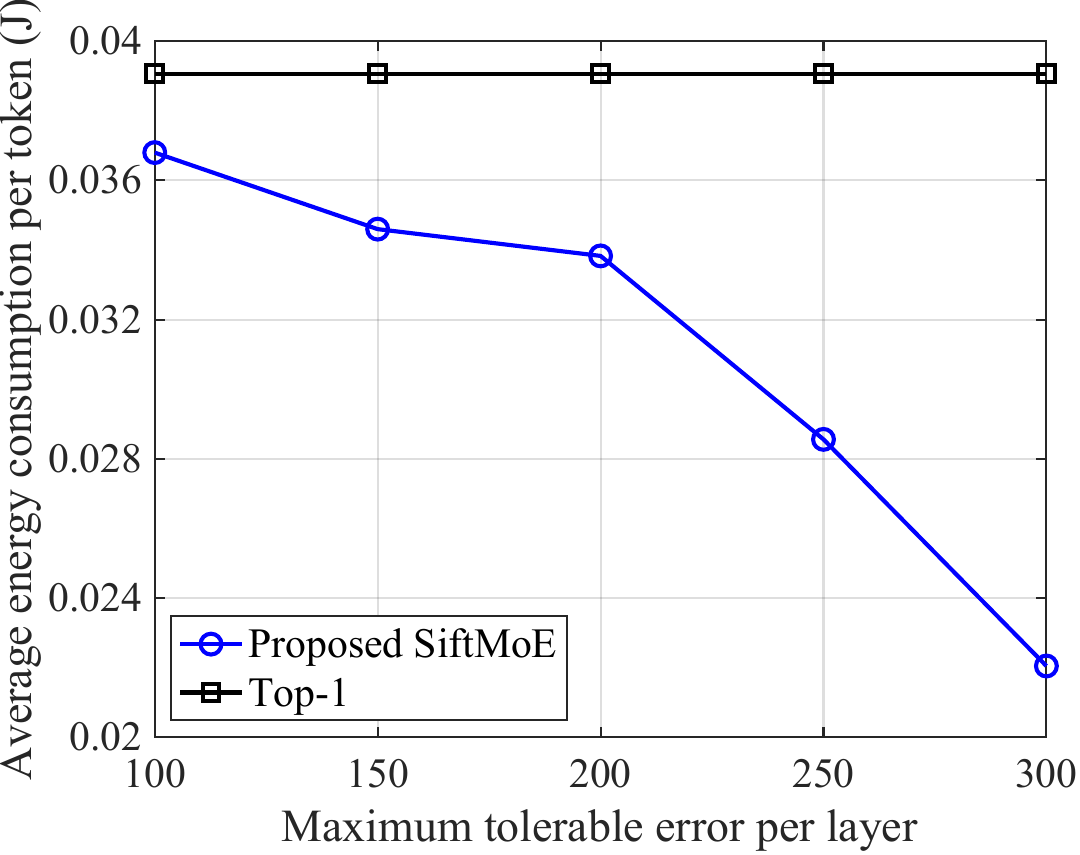}\label{fig:fig_EnergyDelta_switch_slow}}
    \quad 
    \subfigure[Accuracy-energy tradeoff.]{\includegraphics[width =0.25\textwidth]{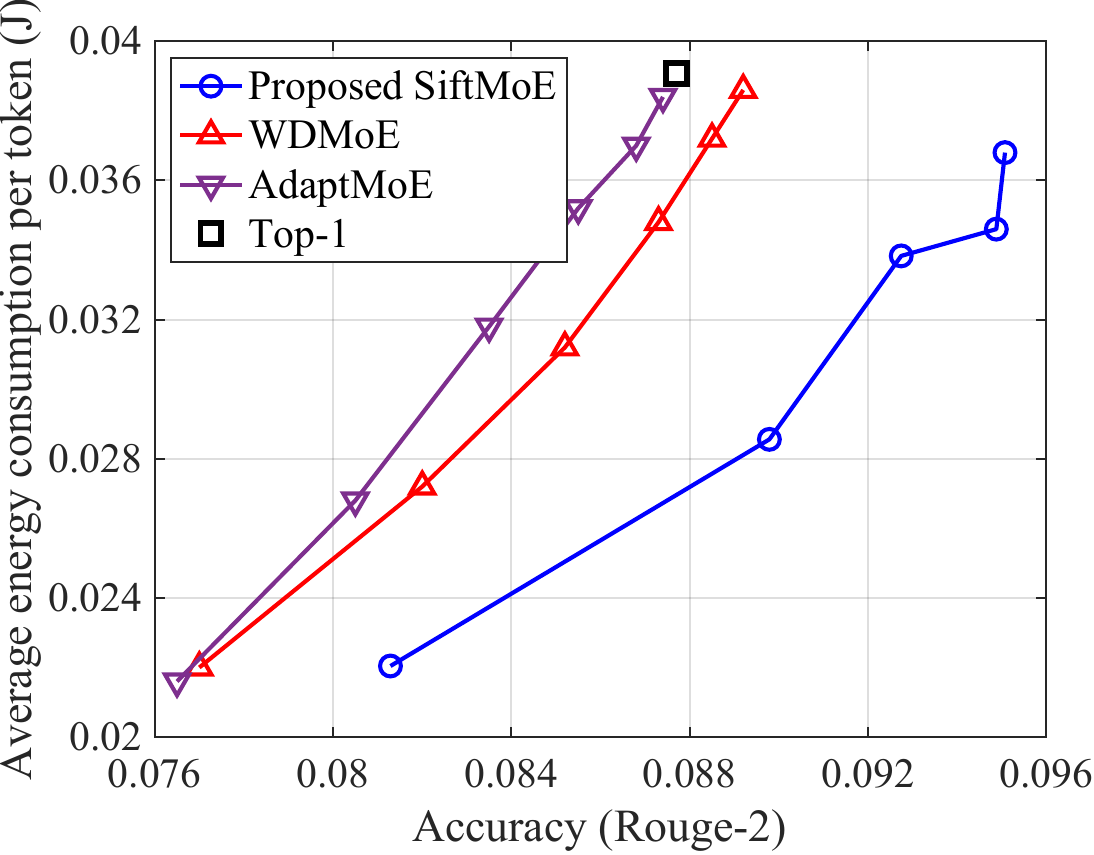}\label{fig:fig_EnergyAccuracy_switch}}
\caption{Performance comparison of different schemes under slow-fading channels for different maximum tolerable errors per layer on Switch-base-8 over XSum, where $B_n = 1$ MHz and $T = 0.7$ s.} 
\label{fig:delta_switch_slow} 
\end{figure*}

\begin{figure*}[!t]
	\centering
	\subfigure[Accuracy.]{\includegraphics[width =0.24\textwidth]{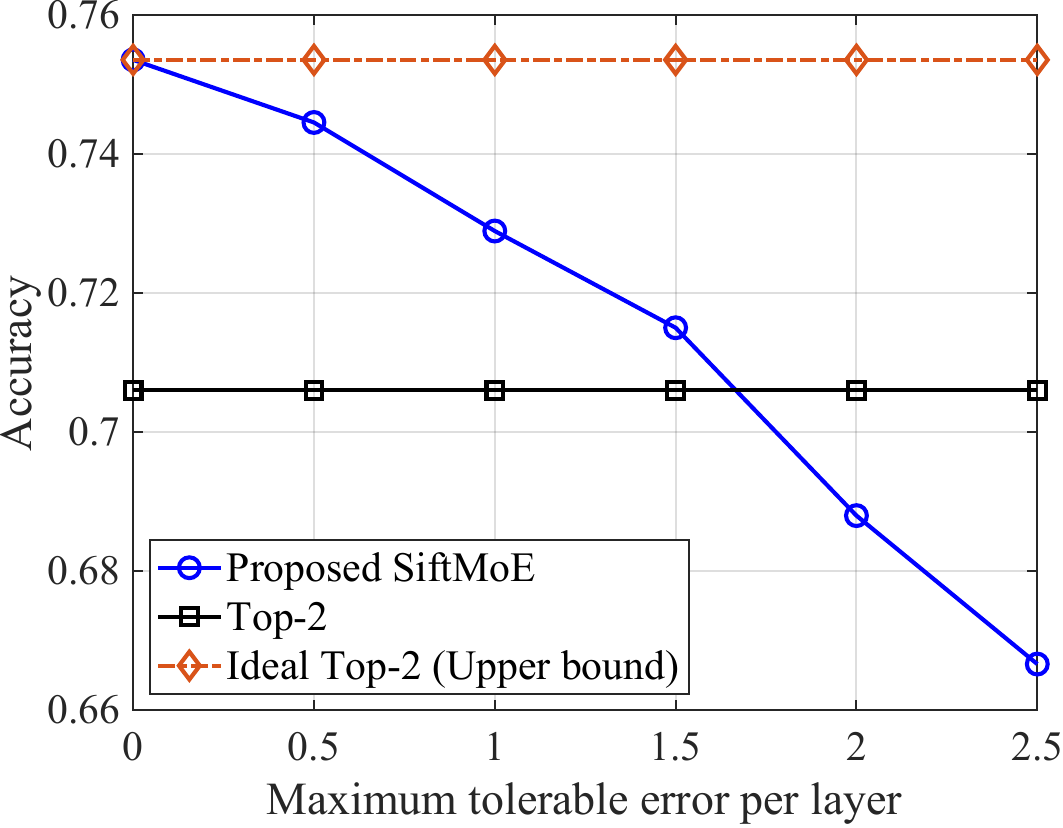}\label{fig:fig_AccDelta_mixtral_slow}}
    \quad 
	\subfigure[Energy consumption per token.]{\includegraphics[width =0.24\textwidth]{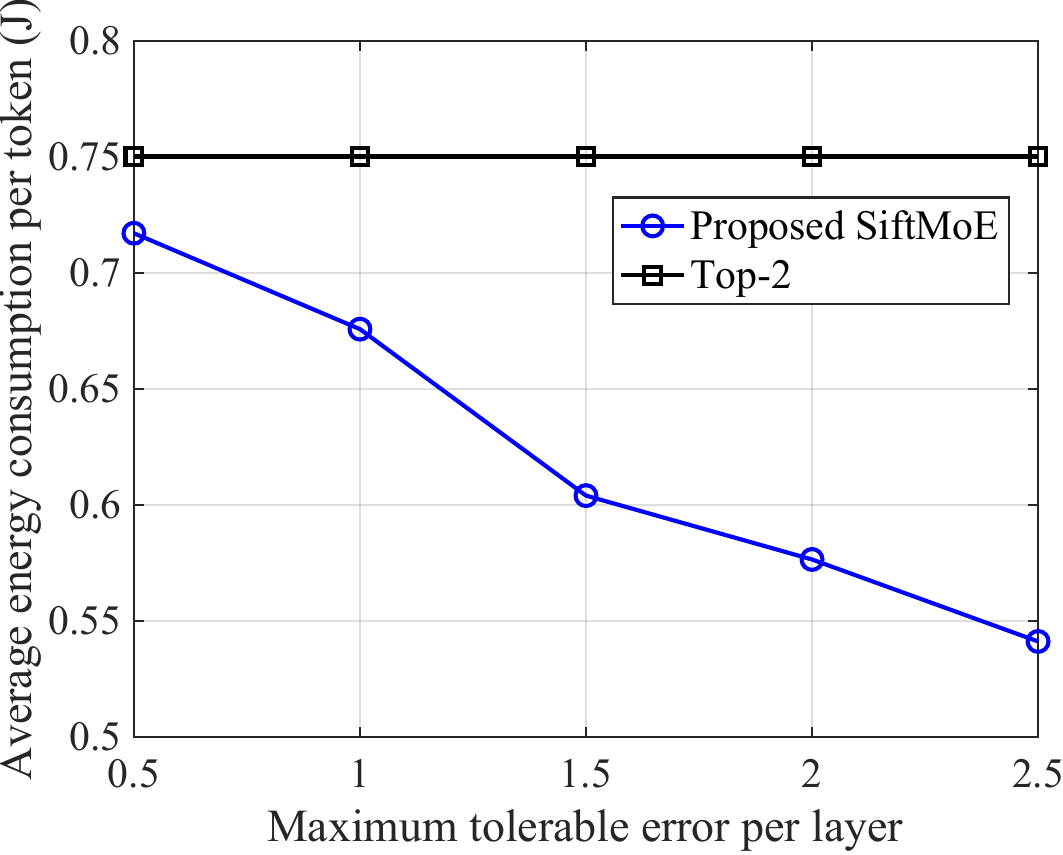}\label{fig:fig_EnergyDelta_mixtral_slow}}
    \quad 
    \subfigure[Accuracy-energy tradeoff.]{\includegraphics[width =0.25\textwidth]{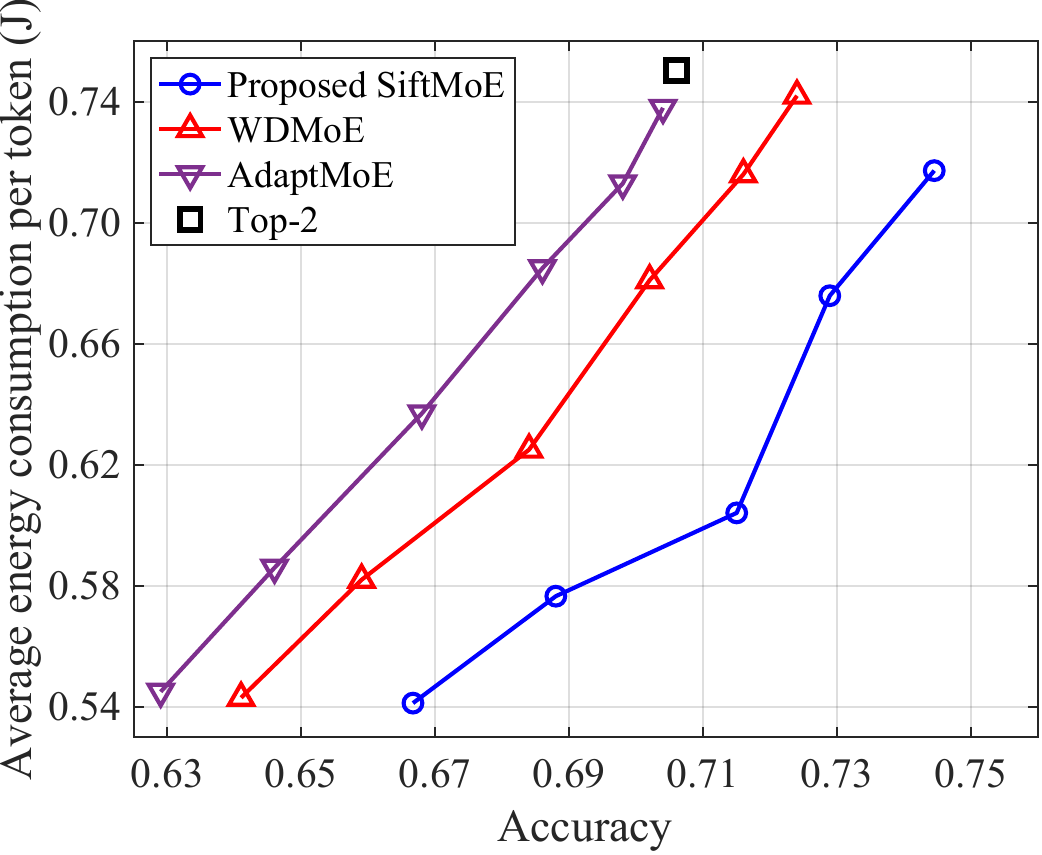}\label{fig:fig_EnergyAccuracy_mixtral}}
\caption{Performance comparison of different schemes under slow-fading channels for different maximum tolerable errors per layer on Mixtral-8$\times$7B over CommonsenseQA, where $B_n = 2$ MHz and $T = 0.074$ s.} 
\label{fig:delta_mixtral_slow}
\end{figure*}

\begin{figure*}[!t]
	\centering
	\subfigure[Accuracy.]{\includegraphics[width =0.25\textwidth]{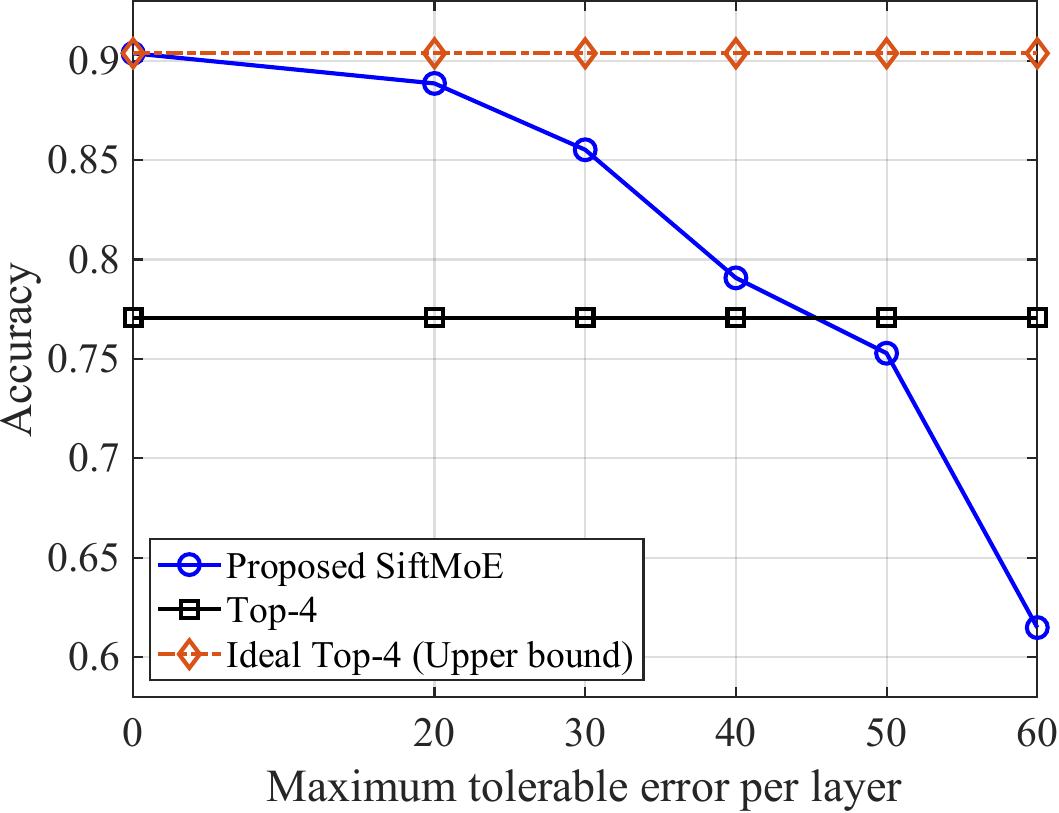}\label{fig:fig_AccDelta_gpt_slow}}
    \quad 
	\subfigure[Energy consumption per token.]{\includegraphics[width =0.25\textwidth]{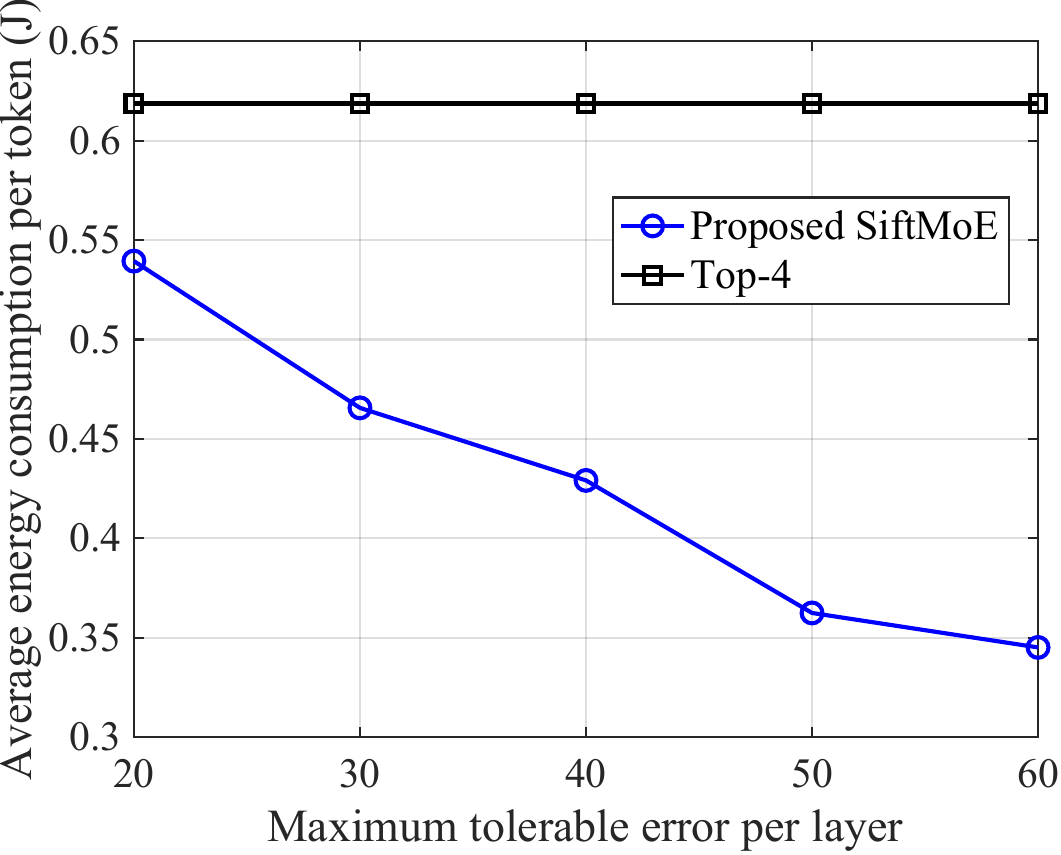}\label{fig:fig_EnergyDelta_gpt_slow}}
    \quad 
    \subfigure[Accuracy-energy tradeoff.]{\includegraphics[width =0.25\textwidth]{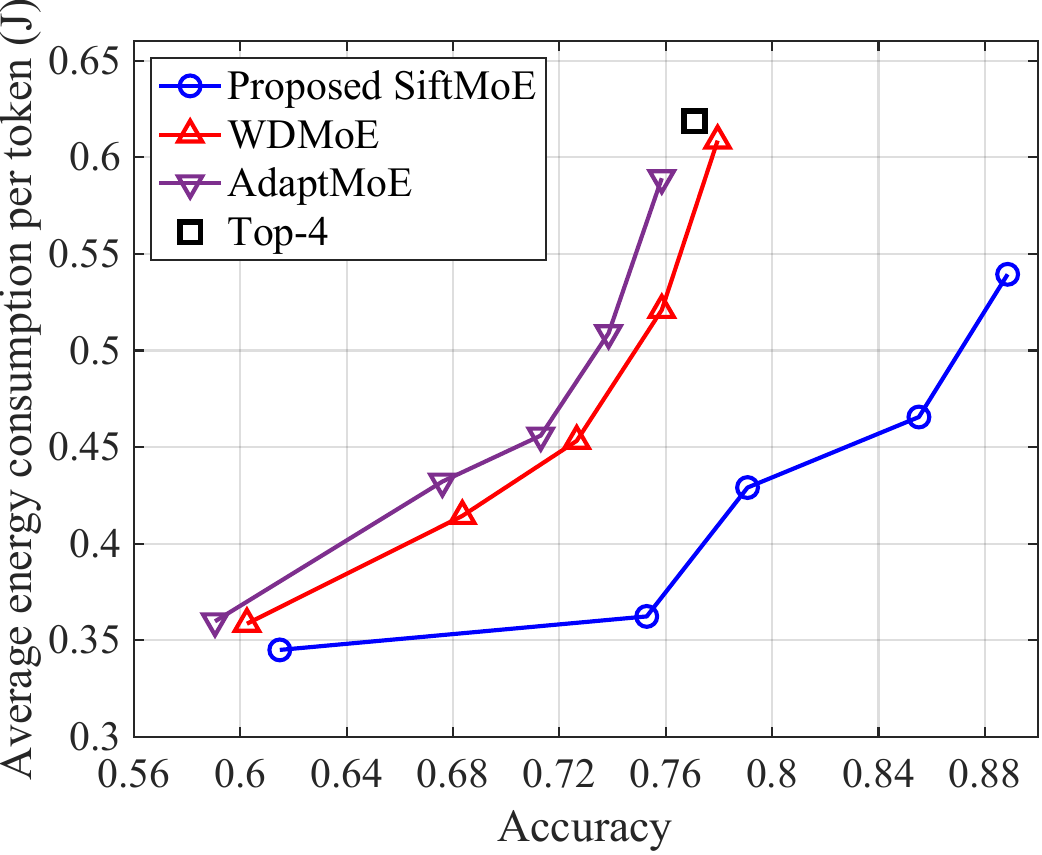}\label{fig:fig_EnergyAccuracy_gpt}}
\caption{Performance comparison of different schemes under slow-fading channels for different maximum tolerable errors per layer on GPT-OSS-20B over GSM8K, where $B_n = 0.8$ MHz and $T = 0.35$ s.} 
\label{fig:delta_gpt_slow} 
\end{figure*}

\subsection{Effects of System Parameters Under Slow Fading}
Following the accuracy-energy simulation results in Section \ref{subsec:accuracy_energy}, the maximum tolerable per-layer error is set to 200 for Switch-base-8, 1.5 for Mixtral-8$\times$7B, and 30 for GPT-OSS-20B in the following simulations, which provide favorable tradeoffs between accuracy preservation and energy reduction. In practice, this parameter can be selected on a validation dataset by evaluating a set of candidate values and choosing the operating point that best matches the target accuracy and energy requirements.

Fig. \ref{fig:energy_comm_switch_slow} illustrates the effect of system constraints on the average energy consumption per token for the Switch-base-8 model over XSum under slow fading. As shown in Fig. \ref{fig:fig_EnergyBand_switch_slow}, when the user-helper bandwidth increases, the average energy consumption per token decreases for all schemes, owing to lower uplink transmission energy. Notably, SiftMoE consistently consumes significantly less energy than the Top-1, WDMoE, and AdaptMoE benchmarks across the entire bandwidth range. This gap is particularly pronounced under limited bandwidth, where the fixed token-helper assignment in the conventional Top-1 scheme and the less flexible expert-dropping strategies in WDMoE and AdaptMoE lead to higher energy consumption.
Fig. \ref{fig:fig_EnergyTime_switch_slow} shows the effect of the per-layer time limit. As the time limit increases, the energy consumption of all schemes decreases, since a looser latency constraint allows the user to transmit hidden states over a longer duration with lower transmission power. SiftMoE still achieves the lowest energy consumption under different latency requirements, demonstrating its robustness to system-level constraints.

Fig.~\ref{fig:energy_comm_mixtral_slow} and Fig.~\ref{fig:energy_comm_gpt_slow} present similar trends for the Mixtral-8$\times$7B model on the CommonsenseQA dataset and the GPT-OSS-20B model on the GSM8K dataset, respectively. For both $K=2$ and $K=4$ cases, SiftMoE consistently achieves lower energy consumption than the corresponding Top-$K$, WDMoE, and AdaptMoE baselines under different bandwidth and latency constraints. These results show that the proposed similarity-aware replacement and skipping strategy can effectively reduce communication energy under different numbers of activated experts in each MoE layer. Overall, these results confirm that SiftMoE remains energy-efficient across different MoE architectures, datasets, routing configurations, and slow-fading system settings.

\begin{figure}[!t]
	\centering
	\subfigure[Effect of bandwidth.]{\includegraphics[width =0.24\textwidth]{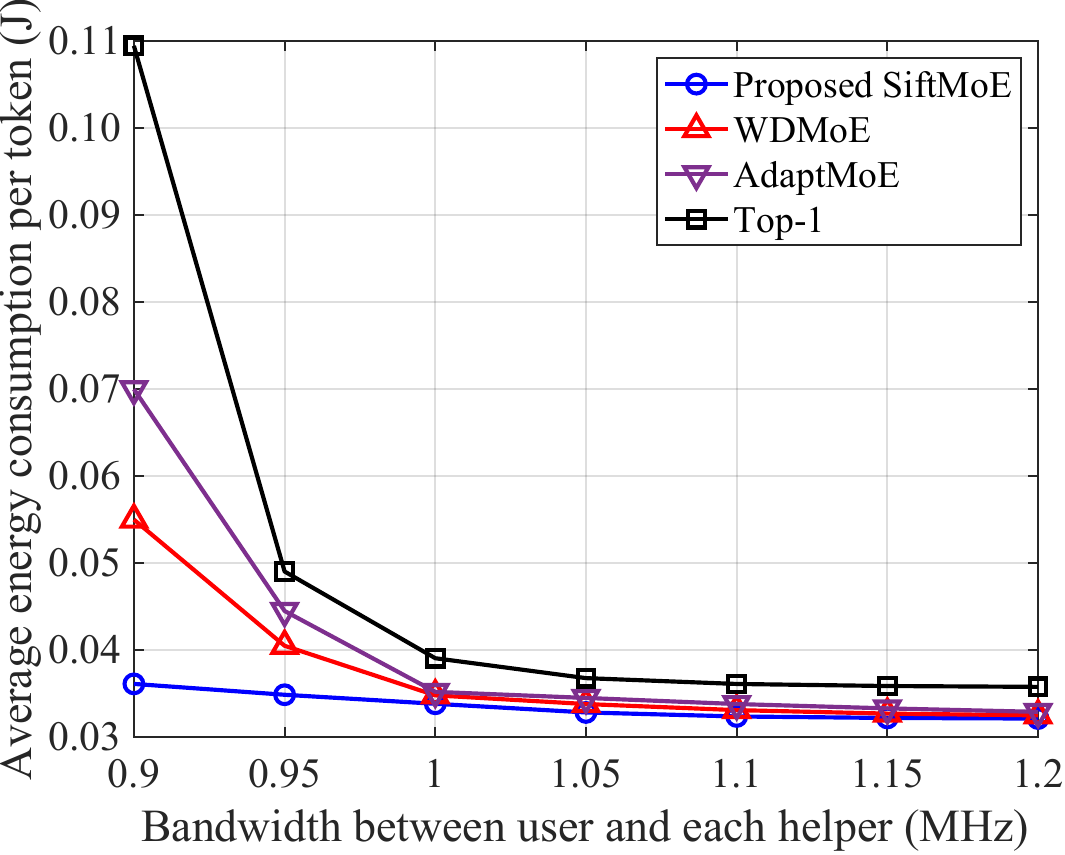}\label{fig:fig_EnergyBand_switch_slow}}
	\subfigure[Effect of time limit per layer.]{\includegraphics[width =0.24\textwidth]{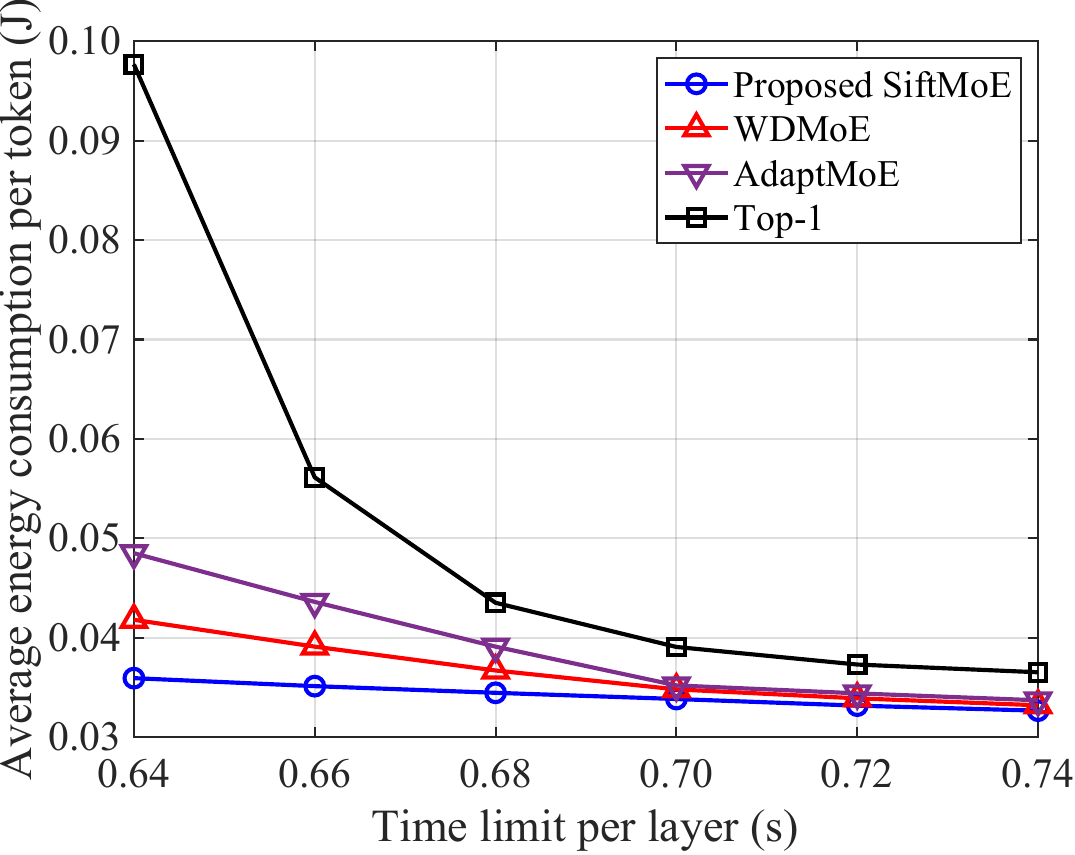}\label{fig:fig_EnergyTime_switch_slow}}
	\caption{Effect of bandwidth and latency requirements on average energy consumption for Switch-base-8 over XSum under slow fading. The default values of $B_n$ and $T$ are set to 1 MHz and 0.7 s.} 
\label{fig:energy_comm_switch_slow} 
\end{figure}

\begin{figure}[!t]
	\centering
	\subfigure[Effect of bandwidth.]{\includegraphics[width =0.235\textwidth]{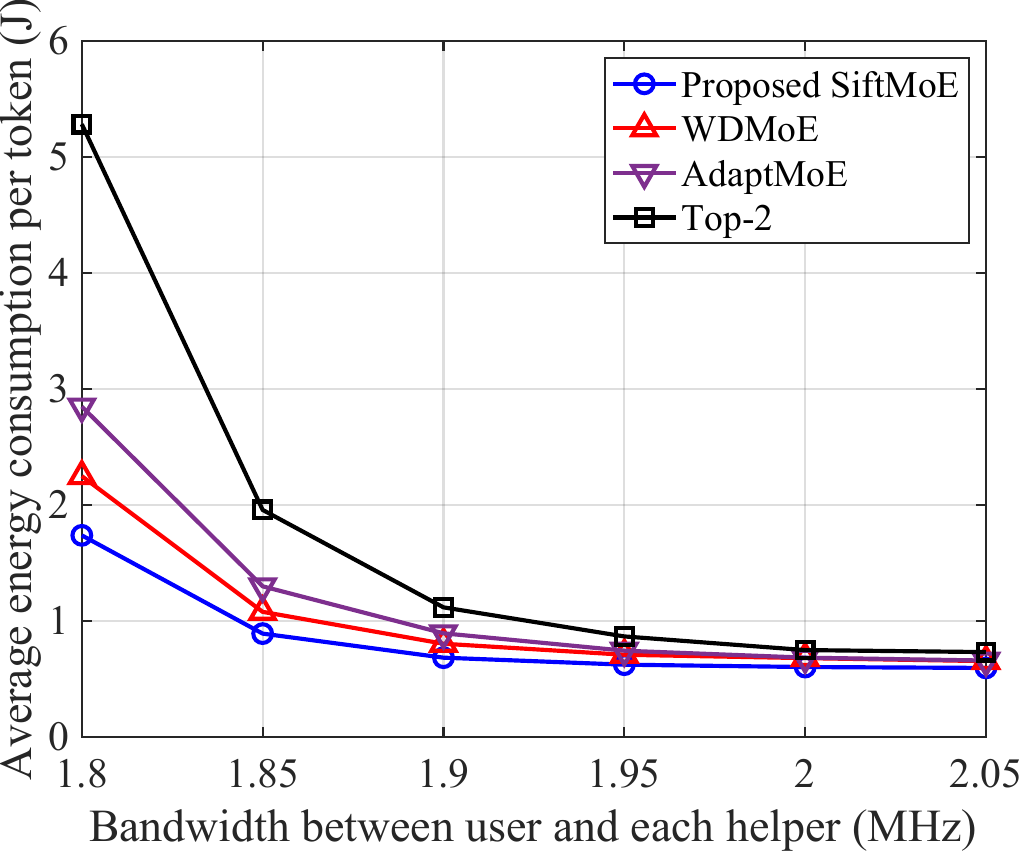}\label{fig:fig_EnergyBand_mixtral_slow}}
	\subfigure[Effect of time limit per layer.]{\includegraphics[width =0.245\textwidth]{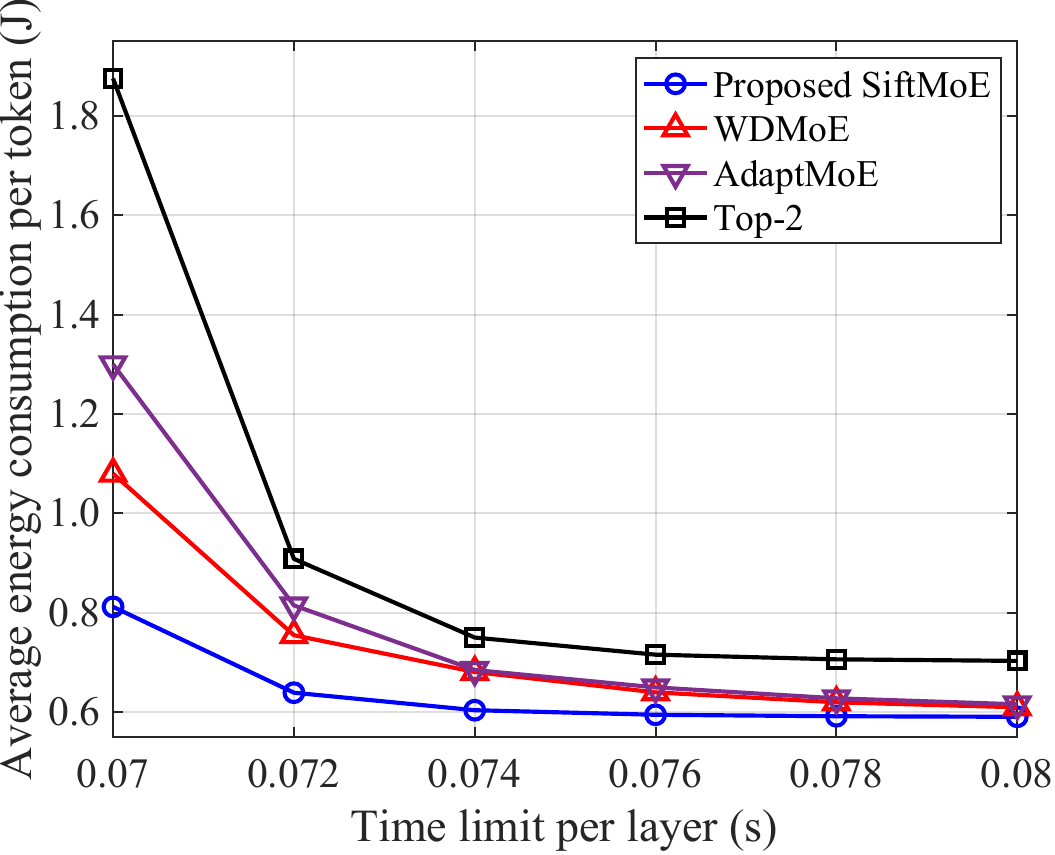}\label{fig:fig_EnergyTime_mixtral_slow}}
	\caption{Effect of bandwidth and latency requirements on average energy consumption for Mixtral-8$\times$7B over CommonsenseQA under slow fading. The default values of $B_n$ and $T$ are set to 2 MHz and 0.074 s.} 
\label{fig:energy_comm_mixtral_slow} 
\end{figure}

\begin{figure}[!t]
	\centering
	\subfigure[Effect of bandwidth.]{\includegraphics[width =0.23\textwidth]{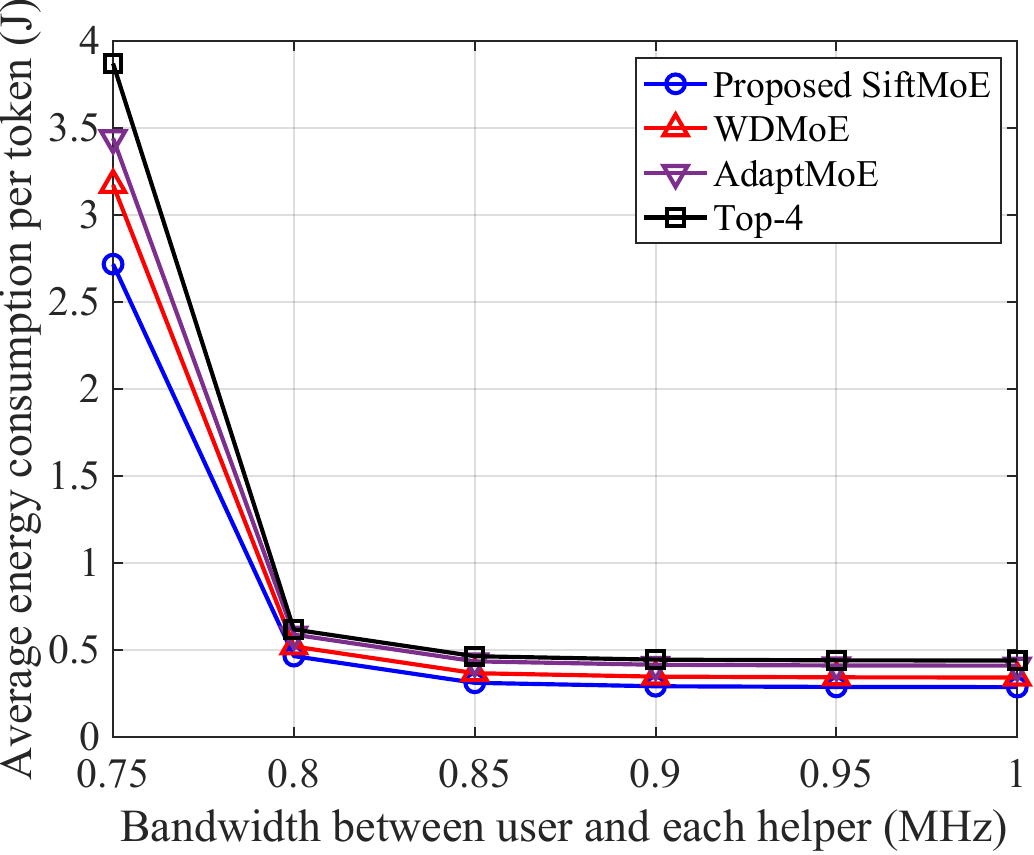}\label{fig:fig_EnergyBand_gpt_slow}}
\subfigure[Effect of time limit per layer.]{\includegraphics[width =0.24\textwidth]{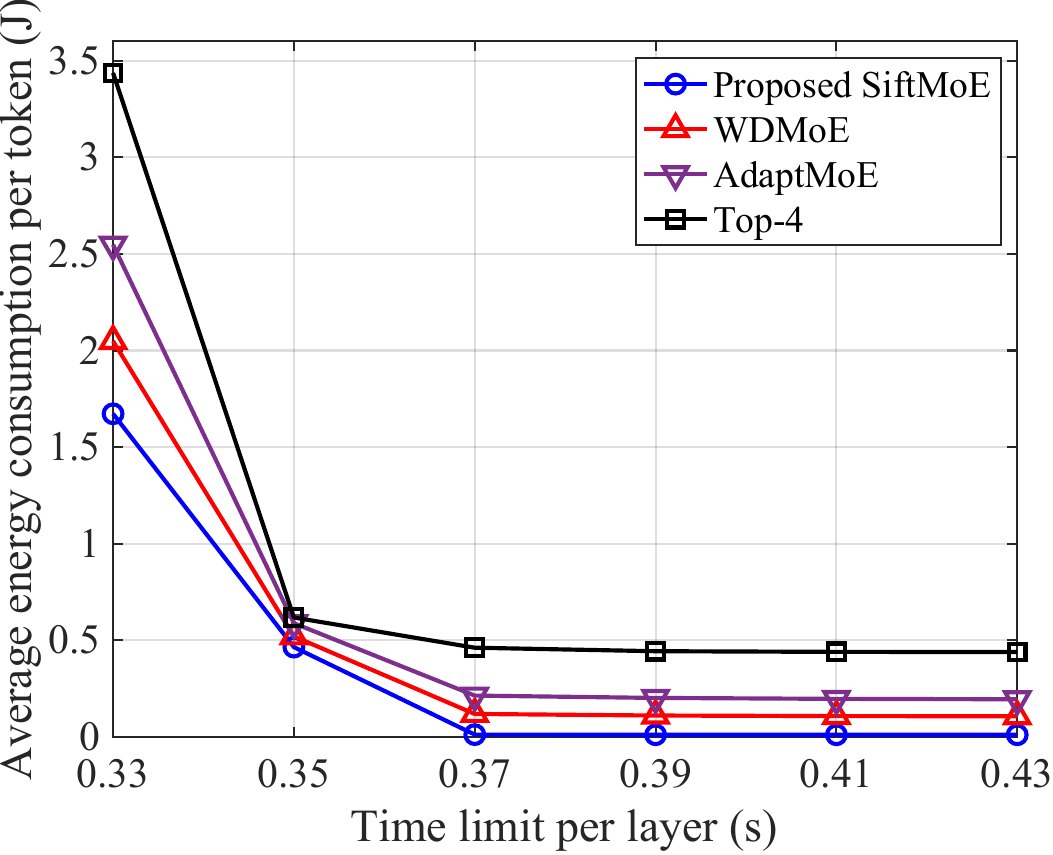}\label{fig:fig_EnergyTime_gpt_slow}}
	\caption{Effect of bandwidth and latency requirements on average energy consumption for GPT-OSS-20B over GSM8K under slow fading. The default values of $B_n$ and $T$ are set to 0.8 MHz and 0.35 s.} 
\label{fig:energy_comm_gpt_slow} 
\end{figure}

\subsection{Effects of System Parameters Under Fast Fading}
We next evaluate the energy performance of SiftMoE under fast fading channels. 
Fig. \ref{fig:energy_comm_switch_fast} shows the effect of system constraints on the average energy consumption for the Switch-base-8 model over XSum under fast fading. 
Here, ``Proposed adaptive” denotes the proposed adaptive transmission strategy, while ``Uniform” denotes the baseline transmission strategy that equally allocates the transmitted data across all time slots.
As shown in Fig. \ref{fig:fig_EnergyBand_switch_fast}, increasing the user-helper bandwidth reduces the energy consumption for all schemes. Compared with the uniform transmission strategy, the proposed adaptive scheme consistently achieves lower energy consumption by exploiting favorable channel realizations to transmit data with lower energy. Moreover, SiftMoE with the proposed adaptive transmission strategy significantly outperforms other baselines across the entire bandwidth range, with greater gains in limited-bandwidth scenarios. 
Similar trends can be observed in Fig. \ref{fig:fig_EnergyTime_switch_fast} with respect to the per-layer time limit, where relaxing the latency constraint further reduces energy consumption, and the proposed dynamic strategy consistently outperforms uniform transmission under fast fading.

Fig. \ref{fig:energy_comm_mixtral_fast} and Fig. \ref{fig:energy_comm_gpt_fast} present the corresponding results for the Mixtral-8$\times$7B model on CommonsenseQA and the GPT-OSS-20B model on GSM8K, respectively. Similar performance trends can be observed in both cases. Overall, these results confirm that under fast fading, combining expert-selection flexibility with channel-aware adaptive transmission is key to achieving energy-efficient WIDE inference.

\begin{figure}[!t]
	\centering
	\subfigure[Effect of bandwidth.]{\includegraphics[width =0.235\textwidth]{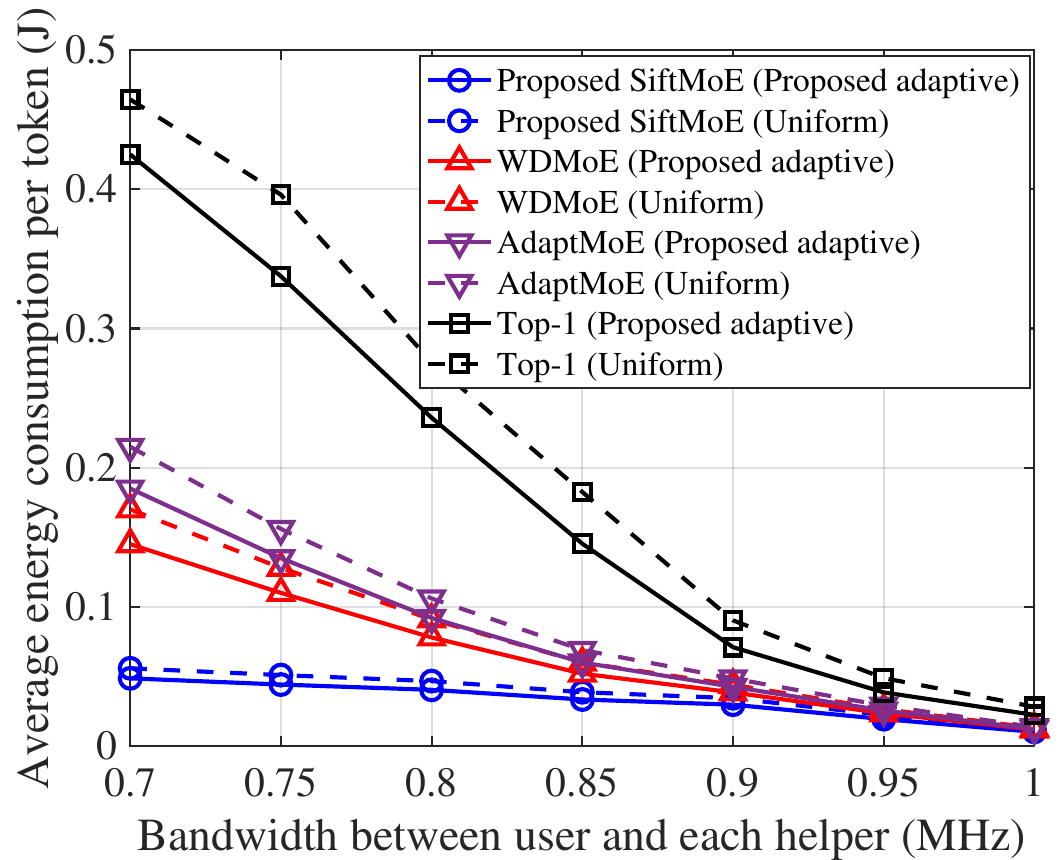}\label{fig:fig_EnergyBand_switch_fast}}
	\subfigure[Effect of time limit per layer.]{\includegraphics[width =0.245\textwidth]{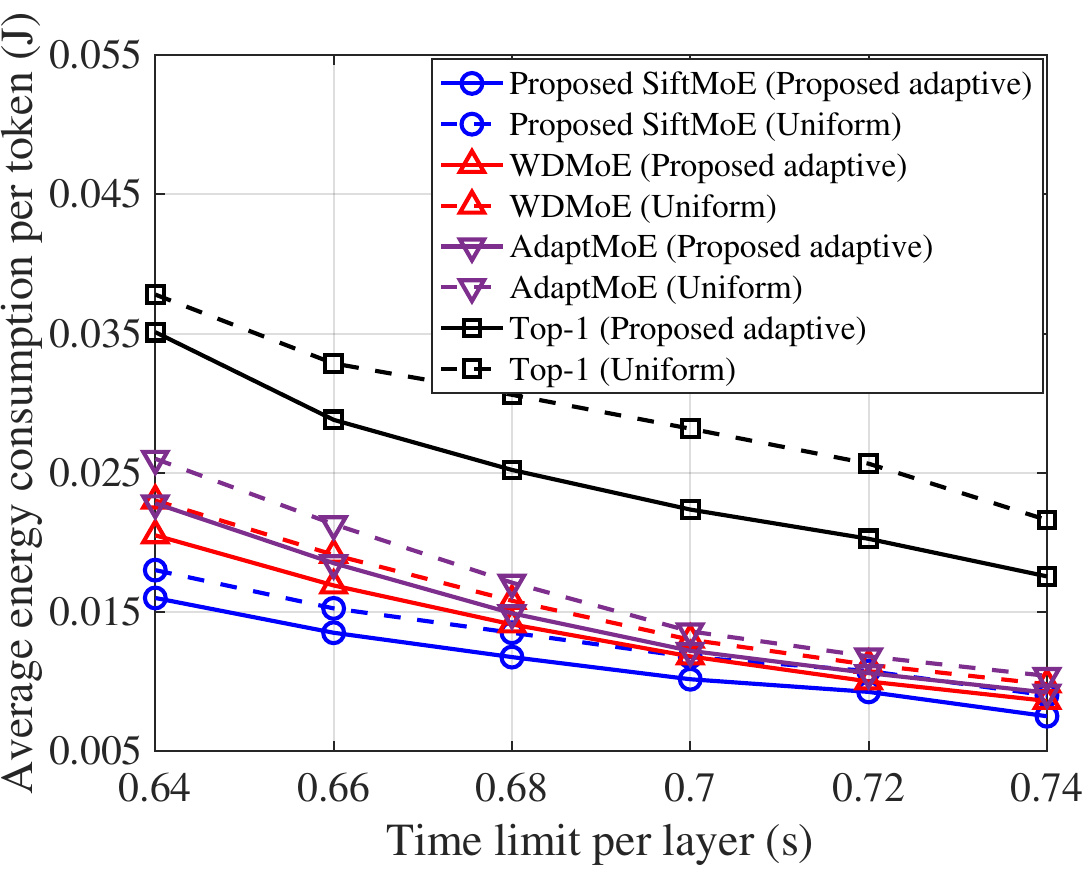}\label{fig:fig_EnergyTime_switch_fast}}
	\caption{Effect of bandwidth and latency requirements on average energy consumption for Switch-base-8 over XSum under fast fading. The default values of $B_n$ and $T$ are set to 1 MHz and 0.7 s.} 
\label{fig:energy_comm_switch_fast} 
\end{figure}

\begin{figure}[!t]
	\centering
	\subfigure[Effect of bandwidth.]{\includegraphics[width =0.235\textwidth]{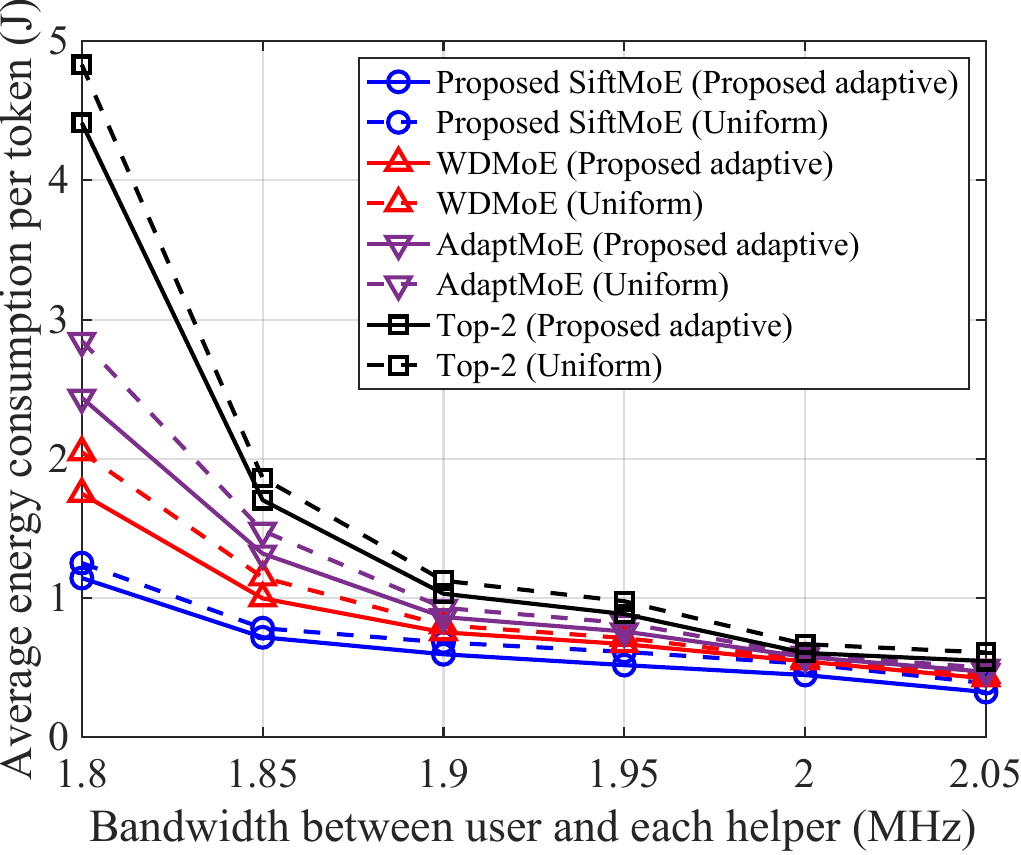}\label{fig:fig_EnergyBand_mixtral_fast}}
	\subfigure[Effect of time limit per layer.]{\includegraphics[width =0.245\textwidth]{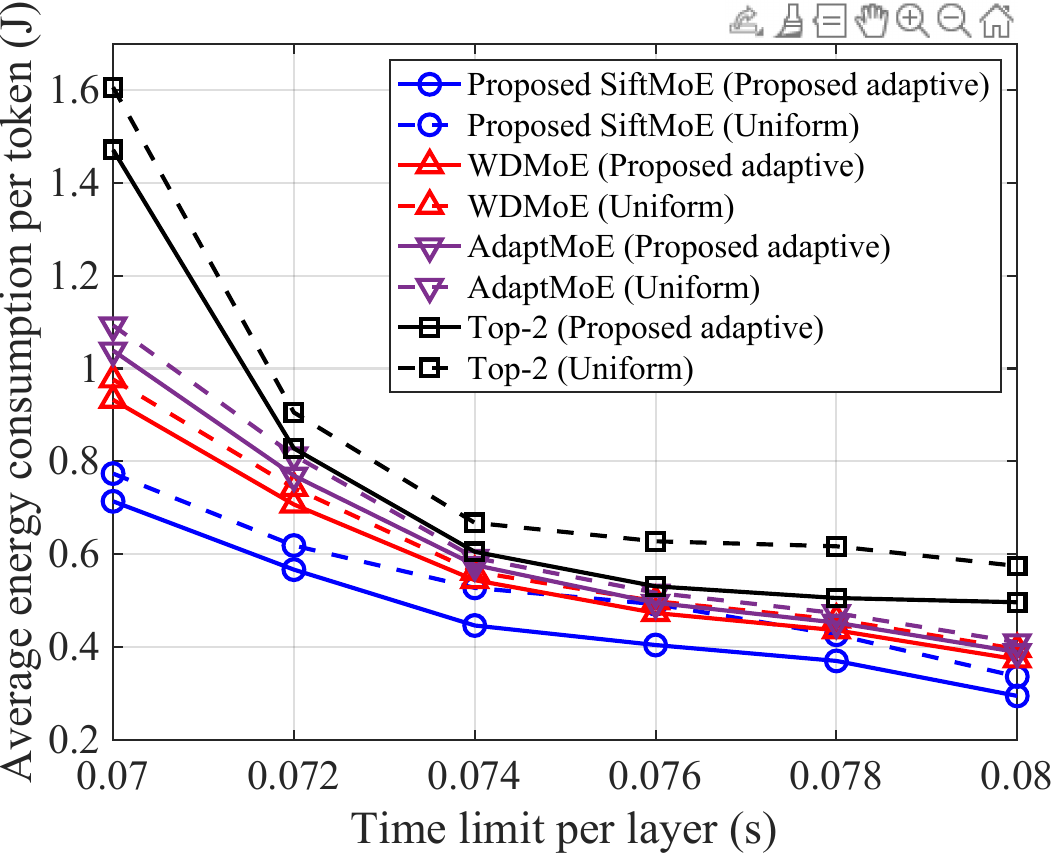}\label{fig:fig_EnergyTime_mixtral_fast}}
	\caption{Effect of bandwidth and latency requirements on average energy consumption for Mixtral-8$\times$7B over CommonsenseQA under fast fading. The default values of $B_n$ and $T$ are set to 2 MHz and 0.074 s.} 
\label{fig:energy_comm_mixtral_fast} 
\end{figure}

\begin{figure}[!t]
	\centering
	\subfigure[Effect of bandwidth.]{\includegraphics[width =0.24\textwidth]{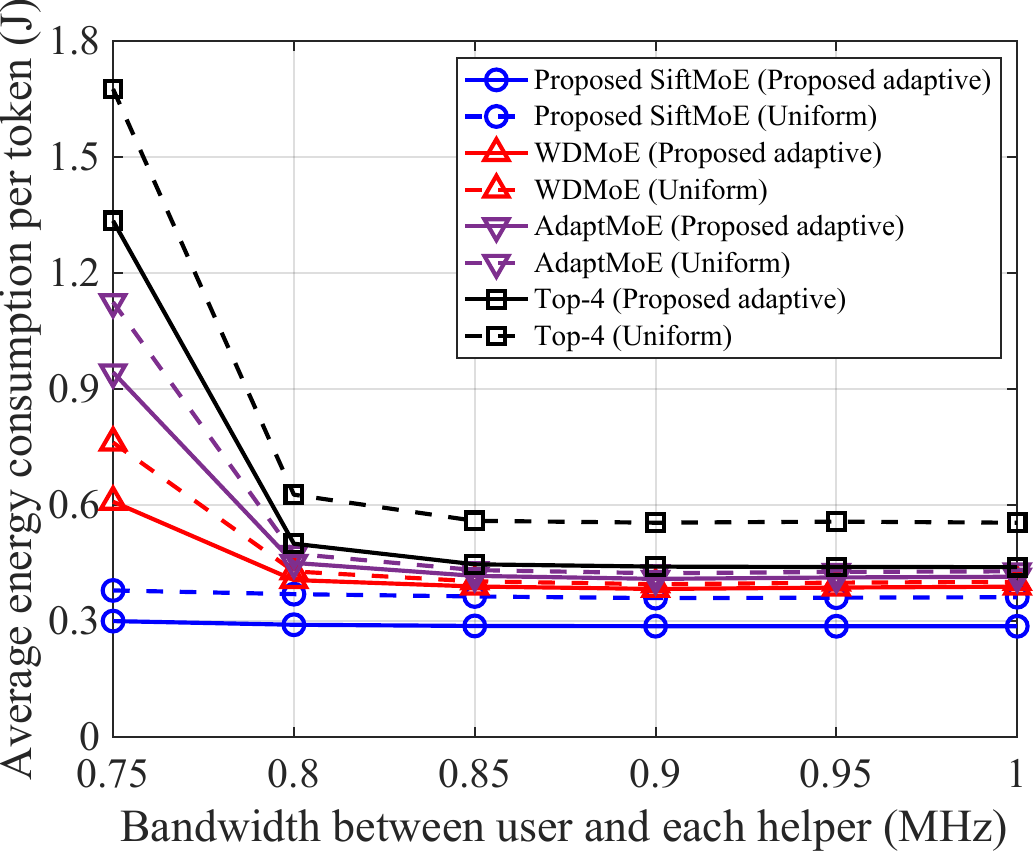}\label{fig:fig_EnergyBand_gpt_fast}}
	\subfigure[Effect of time limit per layer.]{\includegraphics[width =0.24\textwidth]{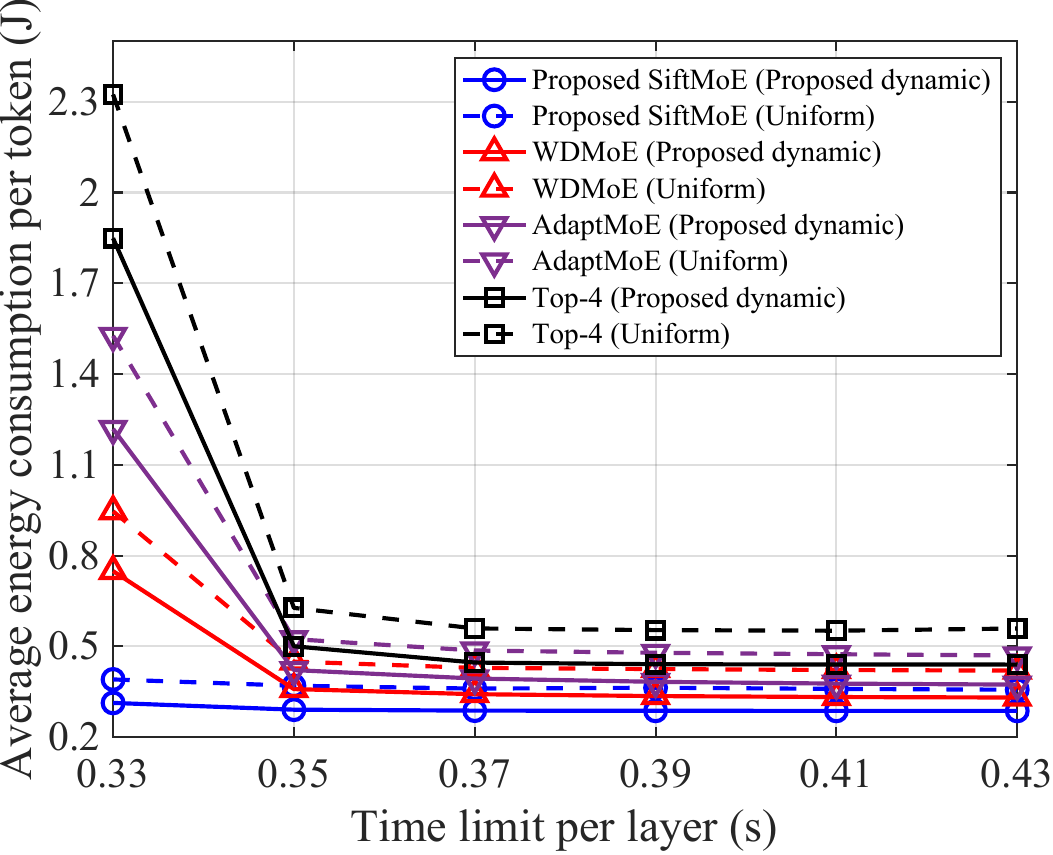}\label{fig:fig_EnergyTime_gpt_fast}}
	\caption{Effect of bandwidth and latency requirements on average energy consumption for GPT-OSS-20B over GSM8K under fast fading. The default values of $B_n$ and $T$ are set to 0.8 MHz and 0.35 s.} 
\label{fig:energy_comm_gpt_fast} 
\end{figure}

	\section{Conclusion Remarks}\label{sec:conclusion}
	In this paper, we revisited distributed MoE inference over wireless edge networks from an energy-efficiency perspective. Rather than strictly adhering to the conventional Top-$K$ routing paradigm, we showed that allowing similarity-aware expert replacement and skipping opens a new research direction for communication-efficient WIDE inference. Our analysis reveals that the impact of expert selection on inference accuracy can be theoretically characterized by layer-wise deviation bounds, thereby enabling control of the accuracy–energy tradeoff. Building on this insight, we developed the SiftMoE framework that integrates expert selection with wireless resource considerations. The results demonstrate that substantial energy savings can be achieved while maintaining controllable accuracy degradation, highlighting the potential of communication-aware expert selection as a practical mechanism for deploying MoE models in wireless edge networks.
    
Several directions deserve further investigation. First, we can extend the current framework to multi-user scenarios, which would introduce additional competition not only among tokens within a user but also across different users sharing communication and computing resources. Second, while this work focuses on expert replacement and skipping within each MoE layer, exploring whether entire MoE layers can be selectively skipped may further reduce communication overhead. Finally, integrating hidden-state transmission with advanced physical-layer techniques such as over-the-air computation may further improve the efficiency of WIDE inference.

\begin{appendices}

\section{Proof of Proposition \ref{prop:layer-expect-perturbation}}\label{proof:prop-layer-expect-perturbation}
By definition, the deviation induced by expert substitution at layer $\ell$ is given by $\tilde{y}^{(\ell)}\left ( z  \right ) - y_{\rm{Top}}^{(\ell)}\left ( z  \right )=
		\sum_{i \in {\mathcal{A}}_{\rm{Top}}^{(\ell)}\left ( z  \right )} g_i^{(\ell)}\left ( z  \right ) \left ( {\rm{FFN}}_{\varphi(i)}^{(\ell)}\left ( z   \right )-{\rm{FFN}}_{i}^{(\ell)}\left ( z   \right ) \right )$.

Taking the $l_2$-norm and applying the triangle inequality, we have
\begin{equation}
    \begin{split}
         \delta^{(\ell)}\left ( z  \right ) & = \left \| \tilde{y}^{(\ell)}\left ( z  \right ) - y_{\rm{Top}}^{(\ell)}\left ( z  \right ) \right \|_2 \\
         &\leq \sum_{i \in {\mathcal{A}}_{\rm{Top}}^{(\ell)}\left ( z  \right )} \left \|g_i^{(\ell)}\left ( z  \right ) \left ( {\rm{FFN}}_{\varphi(i)}^{(\ell)}\left ( z   \right )-{\rm{FFN}}_{i}^{(\ell)}\left ( z   \right ) \right ) \right \|_2 \\
        & = \sum_{i \in {\mathcal{A}}_{\rm{Top}}^{(\ell)}\left ( z  \right )} g_i^{(\ell)}\left ( z  \right ) \left \|  {\rm{FFN}}_{\varphi(i)}^{(\ell)}\left ( z   \right )-{\rm{FFN}}_{i}^{(\ell)}\left ( z   \right ) \right \|_2,
    \end{split}
\end{equation}
where the last equality follows from $g_i^{(\ell)}\left ( z  \right ) \geq 0$.

Then, the deviation magnitude at layer $\ell$ satisfies $\delta^{(\ell)}\left ( z  \right ) \leq \sum_{i \in {\mathcal{A}}_{\rm{Top}}^{(\ell)}\left ( z  \right )} g_i^{(\ell)}\left ( z  \right ) F_{i,\varphi(i)}^{(\ell)}\left ( z   \right ) $.

\section{Proof of Theorem \ref{prop:final_deviation_perlayer}}\label{proof:prop-final_deviation_perlayer}

At layer $\ell$, the discrepancy between the substituted output and the original output can be decomposed into (A) an amplification of the input deviation and (B) the intrinsic error introduced by expert selection at the current layer.
	\begin{equation}\label{equ:el}
		\begin{split}
		e^{(\ell)} &=	 \left \| \tilde{y}^{(\ell)}\left ( \tilde{y}^{(\ell-1)} \right ) -y_{\rm{Top}}^{(\ell)}\left ( y_{\rm{Top}}^{(\ell-1)} \right )  \right \|_2\\
		& = \left\| \underbrace{
\tilde{y}^{(\ell)}\!\left({\tilde{y}}^{(\ell-1)}\right)
-
{\tilde{y}}^{(\ell)}\left(y_{\rm{Top}}^{(\ell-1)}\right)
}_{\text{(A)}} \right. \\
& + \left. \underbrace{
{\tilde{y}}^{(\ell)}\left(y_{\rm{Top}}^{(\ell-1)}\right)
-
y_{\rm{Top}}^{(\ell)} \left(y_{\rm{Top}}^{(\ell-1)}\right)
}_{\text{(B)}} \right\| \\
& \le \left\| {\tilde{y}}^{(\ell)}(\tilde z)-{\tilde{y}}^{(\ell)}(z) \right\|_2 + \delta^{(\ell)},
		\end{split}
	\end{equation}
    where we have set $z\triangleq y_{\rm{Top}}^{(\ell-1)}$ and $\tilde z\triangleq {\tilde{y}}^{(\ell-1)}$.

We first analyze the propagated input deviation term. Using the MoE structure of the substituted layer, the difference ${\tilde{y}}^{(\ell)}(\tilde z) - {\tilde{y}}^{(\ell)}(z)$ can be expanded as the sum of two components: $\tilde{y}^{(\ell)}\left ( \tilde{z} \right ) -\tilde{y}^{(\ell)}\left ( z \right )= \sum_{i \in \mathcal{E}^{(\ell)} } g_i^{(\ell)}\left (  \tilde{z}  \right ) \mathrm{FFN}_i^{(\ell)}\left (  \tilde{z}   \right )-\sum_{i \in \mathcal{E}^{(\ell)} } g_i^{(\ell)}\left (  z  \right ) \mathrm{FFN}_i^{(\ell)}\left ( z   \right )= \sum_{i \in \mathcal{E}^{(\ell)} } \left(g_i^{(\ell)}\left (  \tilde{z}  \right ) - g_i^{(\ell)}\left (  z \right )  \right)\mathrm{FFN}_i^{(\ell)}\left (  \tilde{z}   \right ) + \sum_{i \in \mathcal{E}^{(\ell)} } g_i^{(\ell)}\left ( z  \right ) \left(  \mathrm{FFN}_i^{(\ell)}\left (  \tilde{z}   \right ) - \mathrm{FFN}_i^{(\ell)}\left (  z   \right )  \right)$.
Using the norm property for scalar and applying the triangle inequality, we obtain that
\begin{equation}\label{equ:l2norm_Ta+Tb}
	\begin{split}
		& \left \|  \tilde{y}^{(\ell)}\left ( \tilde{z} \right ) -\tilde{y}^{(\ell)}\left ( z \right )  \right \|_2  \\
        &\leq  
		\underbrace{\sum_{i \in \mathcal{E}^{(\ell)} } \left| g_i^{(\ell)}\left (  \tilde{z}  \right ) - g_i^{(\ell)}\left (  z \right )  \right|   \left\| \mathrm{FFN}_i^{(\ell)}\left (  \tilde{z}   \right )  \right\|_2}_{T_a^{(\ell)}} \\
        & +  \underbrace{\sum_{i \in \mathcal{E}^{(\ell)} } g_i^{(\ell)}\left ( z  \right ) \left \|  \mathrm{FFN}_i^{(\ell)}\left (  \tilde{z}   \right ) - \mathrm{FFN}_i^{(\ell)}\left (  z   \right ) \right \|_2}_{T_b^{(\ell)}}.
	\end{split}
\end{equation}

We first focus on the term $T_a^{(\ell)}$ in (\ref{equ:l2norm_Ta+Tb}), which captures the error induced by the variation of the gating scores under input deviation. Applying the triangle inequality and using Assumption~\ref{assumption:bounded_output}, we factor out the maximum expert output norm and obtain, we obtain the following upper bound, given by
\begin{equation}\label{equ:T_a}
\begin{split}
     T_a^{(\ell)}& \leq \left ( \sum_{i \in \mathcal{E}^{(\ell)} } \left| g_i^{(\ell)}\left (  \tilde{z}  \right ) - g_i^{(\ell)}\left (  z \right )  \right|    \right )  \cdot \max_{i \in \mathcal{E^{(\ell)}} }  \left\| \mathrm{FFN}_i^{(\ell)}\left (  \tilde{z}   \right )  \right\|_2 \\
     & =\left \| \Delta g^{(\ell)}  \right \|_1 \cdot B_{\max}^{(\ell)},
\end{split}
\end{equation}
where $\Delta g^{(\ell)}$ denotes the vector of gating-weight differences with entries $\Delta g_i^{(\ell)} = g_i^{(\ell)}(\tilde z) - g_i^{(\ell)}(z)$. In addition, $\left \| \Delta g^{(\ell)}  \right \|_1$ has the following upper bound: 
$ \left \| \Delta g^{(\ell)}  \right \|_1 \leq \sum_{i \in \mathcal{E}^{(\ell)} } \left( g_i^{(\ell)}\left (  \tilde{z}  \right ) + g_i^{(\ell)}\left (  z \right )  \right)   =  \sum_{i\in \mathcal E^{(\ell)}} g_i^{(\ell)}(\tilde z) + \sum_{i\in \mathcal E^{(\ell)}} g_i^{(\ell)}(z) = 2$.
Therefore, $T_a \leq 2 B_{\max}^{(\ell)}$ holds.

Next, we bound the term $T_b^{(\ell)}$ in (\ref{equ:l2norm_Ta+Tb}), which captures the output deviation caused by passing a perturbed hidden representation through the same set of experts. Since the gating scores satisfy $\sum_{i \in \mathcal{E}^{(\ell)}} g_i^{(\ell)}(z)=1$, with Assumption \ref{assumption:Lip_expert}, $T_b^{(\ell)}$ has the following upper bound: 
\begin{equation}\label{equ:T_b}
    T_b^{(\ell)}  \leq \sum_{i \in \mathcal{E}^{(\ell)} } g_i^{(\ell)}\left ( z  \right ) \beta_{{\rm{E}},i}^{(\ell)}  \left \|  \tilde{z} -z \right \|_2 \leq  \beta_{{\rm E},\max}^{(\ell)}    \left \|  \tilde{z} -z \right \|_2.
\end{equation}

By first substituting (\ref{equ:T_a}) and (\ref{equ:T_b}) into (\ref{equ:l2norm_Ta+Tb}), and then substituting the result into (\ref{equ:el}), we obtain $e^{(\ell)} \leq \beta_{{\rm E},\max}^{(\ell)} e^{(\ell-1)} +2 B_{\max}^{(\ell)} +\delta^{(\ell)}$.
By setting $e^{(0)}=0$ and unrolling this recursion over the first $\ell$ layers, we obtain the following bound on the output deviation at layer $\ell$: $e^{(L)} \le \sum_{r=1}^{\ell}\left(2 B_{\max}^{(r)}+\delta^{(r)}\right) \prod_{t=r+1}^{\ell} \beta_{{\rm E},\max}^{(t)}$.

\section{Proof of Proposition \ref{prop:discrete-convex}}\label{proof:prop-discrete-convex}
From (\ref{equ:energy_local}), it is straightforward to observe that $E_{\rm{UE}}^{(\ell)}$ increases linearly with $D_{\rm{UE}}^{(\ell)}$. Hence, we focus on characterizing the relationship between the helper-side energy consumption $E_n^{(\ell)}$ and $D_n^{(\ell)}$.

Although $D_n^{(\ell)}$ is an integer-valued decision variable in the considered subproblem, we first relax it to a continuous variable $D_n^{(\ell)} \in \mathbb{R}_+$ and study the convexity of $E_n^{(\ell)}\left ( D_n^{(\ell)} \right ) $ over the continuous domain. This relaxation is introduced solely for analytical purposes. 

For notation simplicity, we define $\gamma_n^{(\ell)} = 
\frac{\phi}{C_n} + \frac{b}{R_{n,{\mathrm{UE}}}^{(\ell)}}$. 
Let $u_n^{(\ell)} = T-\gamma_n^{(\ell)} D_n^{(\ell)}$, then the objective function can be rewritten as a composite function $E_n^{(\ell)}(D_n^{(\ell)}) = \frac{N_0 B_n }{d_n^{-\alpha}G_{\mathrm{UE},n}h_{n}^{(\ell)}}\psi(u_n^{(\ell)})$, where
$\psi(u_n^{(\ell)}) = \left [ 2^{\frac{b}{B_n\gamma_n^{(\ell)}}\left ( \frac{T}{u_n^{(\ell)}}-1  \right )  } -1\right ] u_n^{(\ell)} = 2^{-\frac{b}{B_n\gamma_n^{(\ell)}}} e^{\frac{bT \ln 2}{B_n\gamma_n^{(\ell)} u_n^{(\ell)}}} u_n^{(\ell)}-u_n^{(\ell)}$.

We first analyze the properties of $\psi(u_n^{(\ell)})$. The first-order derivative of $\psi(u_n^{(\ell)})$ is given by $\psi'(u_n^{(\ell)}) = 2^{-\frac{b}{B_n\gamma_n^{(\ell)}}}e^{\frac{bT \ln 2}{B_n\gamma_n^{(\ell)} u_n^{(\ell)}}}\left ( 1-\frac{bT \ln 2}{B_n\gamma_n^{(\ell)} u_n^{(\ell)}} \right )-1 $. To characterize its sign, we define an auxiliary function $f(x) = e^x (1-x)$. Its derivative satisfies $f'(x) = -x e^x < 0$ for all $x>0$, implying that $f(x)$ is strictly decreasing when $x>0$ and hence $f(x)<f(0)=1$. Consequently, $e^{\frac{bT \ln 2}{B_n\gamma_n^{(\ell)} u_n^{(\ell)}}}\left ( 1-\frac{bT \ln 2}{B_n\gamma_n^{(\ell)} u_n^{(\ell)}} \right ) <1$. Since $2^{-\frac{b}{B_n\gamma_n^{(\ell)}}} <1$, it follows that $\psi'(u_n^{(\ell)}) <0$ for all $u_n^{(\ell)}>0$.

Moreover, the second-order derivative of $\psi(u_n^{(\ell)})$ is $\psi''(u_n^{(\ell)}) = 2^{-\frac{b}{B_n\gamma_n^{(\ell)}}}e^{\frac{bT \ln 2}{B_n\gamma_n^{(\ell)} u_n^{(\ell)}}}\left ( \frac{bT \ln 2}{B_n\gamma_n^{(\ell)} } \right )^2 \frac{1}{\left ( u_n^{(\ell)} \right ) ^3}  $, which is strictly positive for all $u_n^{(\ell)}>0$. Therefore, $\psi(u_n^{(\ell)})$ is a decreasing and convex function with respect to $u_n^{(\ell)}$.

Next, we examine the monotonicity and convexity of the composite function $\psi(u(D_n^{(\ell)}))$, i.e., $E(D_n^{(\ell)})$. Note that $u'\left ( D_n^{(\ell)} \right ) = -\gamma_n^{(\ell)} <0$ and $u''\left ( D_n^{(\ell)} \right ) =0$. By the chain rule, the first-order derivative of $E_n^{(\ell)}(D_n^{(\ell)})$ is $\frac{\mathrm{d} E_n^{(\ell)}\left ( D_n^{(\ell)} \right ) }{\mathrm{d} D_n^{(\ell)}} = \psi'\left ( u\left ( D_n^{(\ell)} \right )  \right ) \cdot u'\left ( D_n^{(\ell)} \right ) > 0$. Furthermore, the second-order derivative satisfies $\frac{\mathrm{d}^2 E_n^{(\ell)}\left ( D_n^{(\ell)} \right ) }{\mathrm{d} \left ( D_n^{(\ell)} \right ) ^2} = \psi''\left ( u\left ( D_n^{(\ell)} \right )  \right ) \cdot \left ( u'\left ( D_n^{(\ell)} \right )  \right )^2 + \psi'\left ( u\left ( D_n^{(\ell)} \right )  \right ) \cdot u''\left ( D_n^{(\ell)} \right ) >0 $.  Hence, $E(D_n^{(\ell)})$ is an increasing and convex function of $D_n^{(\ell)}$ over the continuous domain $D_n^{(\ell)}\in\mathbb{R}_+$.

Finally, when $D_n^{(\ell)}$ is restricted to integer values with unit spacing, the discrete sequence $\{E_n^{(\ell)}(D_n^{(\ell)})\}_{D_n^{(\ell)}\in\mathbb{Z}_+}$ preserves these properties. Specifically, the monotonicity of $E_n^{(\ell)}(D_n^{(\ell)})$ implies $E_n^{(\ell)}(D_n^{(\ell)}+1)-E_n^{(\ell)}(D_n^{(\ell)}) > 0$, and the convexity of $E_n^{(\ell)}(D_n^{(\ell)})$ over the continuous domain further guarantees that the second-order finite difference is nonnegative, i.e., $E_n^{(\ell)}(D_n^{(\ell)}+1)-E_n^{(\ell)}(D_n^{(\ell)}) > E_n^{(\ell)}(D_n^{(\ell)}) - E_n^{(\ell)}(D_n^{(\ell)}-1)$ for any $D_n^{(\ell)} \geq 1$. Therefore,  $E_n^{(\ell)}$ is a non-decreasing discrete convex function with respect to $D_n^{(\ell)}$.

\section{Proof of Theorem \ref{theorem:fastfading}}\label{proof:theorem-fastfading}
	 We take the transmission strategy to a specific helper $n$ and an MoE layer $\ell$ as an example. Thus, we omit the index of helper $n$ and layer $\ell$ and only leave the index of time slot $t$.
Then, we have $ {\bar{U}}_Q \left( \beta_Q \right)  = \mathbb{E}_h \left [ U_{Q} \left(  \beta_Q , h_Q      \right) \right ]  =   c\left ( 2^{a\beta_Q} -1 \right ) \mathbb{E} \left [ \frac{1}{ h_Q  }\right ]  $.

The optimal transmitted bits in slot $(Q-1)$, i.e., $\gamma_{Q- 1}^{\ast}$, is given by $\gamma_{Q- 1}^{\ast} = \arg \min \limits_{\gamma_{Q- 1}}  \frac{c\left ( 2^{a\gamma_{Q- 1}} -1 \right )}{h_{Q-1} }+ c\left [ 2^{a\left ( \beta_{Q-1}-\gamma_{Q-1} \right ) } -1 \right ] \mathbb{E}\left [ \frac{1}{ h_Q  } \right ]= \frac{1}{2a}\left [ a \beta_{Q-1} + \log_2\left ( h_{Q-1} \mathbb{E}\left [ \frac{1}{ h_Q  } \right ] \right )  \right ]$, and the cost-to-go function is $U_{Q-1}\left(\beta_{Q-1},h_{Q-1}  \right)= c\left [ 2\cdot 2^{\frac{a \beta_{Q-1}}{2} } \left ( \frac{1}{h_{Q-1}  } \mathbb{E}\left [ \frac{1}{ h_Q  } \right ]   \right )^{\frac{1}{2} } -\left ( \frac{1}{h_{Q-1}}+  \mathbb{E}\left [ \frac{1}{ h_Q  } \right ]\right )  \right ] $.

    	Calculating the above repeatedly, we can obtain the optimal transmitted bits as $\gamma_{Q-i}^{\ast} = \frac{1}{\left ( i+1 \right )a } \left [ a \beta_{Q-i}+\sum_{j=0}^{i-1} \log_2\left ( h_{Q-i} \mathbb{E}\left [ \frac{1}{ h_{Q-j}  } \right ] \right )   \right ] $. Then, $U_{Q-i}^{\ast}= c \left[ \left ( i+1 \right )2^{\frac{a \beta_{Q-i}}{i+1} }
		\left ( \frac{1}{h_{Q-i}}\prod_{j=0}^{i-1 } \mathbb{E}\left [ \frac{1}{ h_{Q-j}  } \right ] \right )^{\frac{1}{i+1} } \right.$ $\left. -\left ( \frac{1}{h_{Q-i}} + \sum_{j=0}^{i-1 } \mathbb{E}\left [ \frac{1}{ h_{Q-j}  } \right ] \right ) \right]$ holds.
    Let $Q-i=1$ and $\beta_{Q-i}=D$, the proof is completed.
    
\end{appendices}	
    
    	\bibliographystyle{IEEEtran}
    \bibliography{IEEEabrv,body/reference.bib}
	
	\ifCLASSOPTIONcaptionsoff
	\newpage
	\fi

\end{document}